\pdfoutput=1
\documentclass[onecolumn,12pt]{IEEEtran}

\usepackage{amsmath,amssymb,amsthm,mathtools,dsfont}
\usepackage{graphicx}
\usepackage{tikz}
\usetikzlibrary{arrows.meta,positioning}
\usepackage{booktabs}
\usepackage[colorlinks=true,allcolors=blue]{hyperref}
\usepackage[capitalise]{cleveref}

\theoremstyle{plain}
\newtheorem{theorem}{Theorem}
\newtheorem{lemma}[theorem]{Lemma}

\newtheorem{proposition}[theorem]{Proposition}
\theoremstyle{definition}
\newtheorem{definition}[theorem]{Definition}
\newtheorem{remark}[theorem]{Remark}

\DeclareMathOperator*{\argmax}{arg\,max}

\DeclareMathOperator*{\esssup}{ess\,sup}

\newcommand{\cC}{{\mathcal C}}
\newcommand{\cD}{{\mathcal D}}

\newcommand{\cN}{{\mathcal N}}\newcommand{\cP}{{\mathcal P}}

\newcommand{\cU}{{\mathcal U}}
\newcommand{\cX}{{\mathcal X}}\newcommand{\cY}{{\mathcal Y}}

\newcommand{\mR}{{\mathbb R}}

\newcommand{\bx}{{\mathbf x}}\newcommand{\by}{{\mathbf y}}
\newcommand{\bC}{{\mathbf C}}

\newcommand{\BRA}[1]{\left( #1 \right)}
\newcommand{\BRAb}[1]{\left[ #1 \right]}
\newcommand{\BRAi}[1]{\left\langle #1 \right\rangle}
\newcommand{\BRAs}[1]{\left\{ #1 \right\}}
\newcommand{\abs}[1]{\left| #1 \right|}
\newcommand{\norm}[1]{\lVert #1 \rVert}

\newcommand{\PR}[1]{\mathbb{P}\left\{ #1 \right\}}
\newcommand{\PRs}[2]{\mathbb{P}_{#1}\left\{ #2 \right\}}
\newcommand{\E}[1]{\mathbb{E}\left( #1 \right)}
\newcommand{\Es}[2]{\mathbb{E}_{#1}\left( #2 \right)}

\newcommand{\Unif}[1]{\cU\left( #1 \right)}
\newcommand{\uU}{\Unif{[0,1]}}
\newcommand{\Ind}[1]{\mathds{1}_{\BRAs{#1}}}

\newcommand{\W}{W_{Y|X}}

\newcommand{\BETA}[3]{\beta_{#1}\!\Big( #2\; ; \; #3\Big)}
\newcommand{\PEP}[2]{p_e\BRA{#1 \mid #2}}
\newcommand{\PEPs}[1]{p_e\BRA{#1}}
\newcommand{\PEPU}[3]{p_e\BRA{#1 \mid #2; #3}}
\newcommand{\PEC}[2]{p_c\BRA{#1 \mid #2}}
\newcommand{\PECU}[3]{p_c\BRA{#1 \mid #2; #3}}
\newcommand{\Fspec}[2]{F\!\BRA{#1\,;\,#2}}
\newcommand{\idens}[2]{\imath\!\BRA{#1\,;\,#2}}

\newcommand{\RCUp}{\mathrm{RCU}^{+}}
\newcommand{\verdu}{Verd\'{u}}

\begin{document}

\title{A Pairwise-Error-Probability Framework for One-Shot Information Theory}

\author{Nir~Elkayam~and~Meir~Feder%
\thanks{The authors are with the Department of Electrical Engineering--Systems,
Tel Aviv University, Tel Aviv, Israel (e-mail: nir.elkayam@gmail.com; meir@eng.tau.ac.il).}}

\maketitle

\begin{abstract}
We develop a one-shot (finite-blocklength) channel-coding framework based on the
pairwise error probability (PEP) of a decoder with randomized tie-breaking. The
tie-breaking rule yields a probability-integral-transform identity: the induced
error spectrum describes both random-coding achievability and exact fixed-code
converse statements, for an arbitrary decoding metric. We derive two variational
identities for metric-weighted tail functionals of the PEP, one through the
Neyman--Pearson $\beta$-functional and one through a reverse channel, valid for
an arbitrary metric. Under matched maximum-likelihood decoding they specialize
to representations of the spectrum itself, which is then jointly convex in the
testing level and the input prior; combined with the reverse-channel
representation, this
gives a linear program for the prior-optimized minimax meta-converse ---
finite-dimensional in general and, for memoryless channels with fixed
alphabets, of size polynomial in the blocklength after a type reduction. Prior
optimization of the random-coding bound is formulated as a concave program over
input distributions with an explicit gradient, solved by a direct first-order
method. The framework recovers several classical one-shot bounds, including the
random-coding union bound and minimax meta-converse of Polyanskiy--Poor--\verdu,
the information-spectrum bounds of Han--\verdu, and the linear-programming
converse of Matthews. Numerical examples on the AWGN and binary Z-channels
illustrate the achievability--converse comparison and the effect of prior
optimization.
\end{abstract}

\begin{IEEEkeywords}
One-shot information theory, finite blocklength, pairwise error probability,
meta-converse, prior optimization, linear programming, random coding,
reverse channel.
\end{IEEEkeywords}

\section{Introduction}\label{sec:intro}

Classical information theory characterizes the limits of communication and
compression in the asymptotic regime, where the blocklength $n\to\infty$.
Modern practice, and much of recent theory, operates at finite (often short)
blocklengths, where the asymptotic limits are not directly informative.  The
non-asymptotic, or \emph{one-shot}, viewpoint brought to prominence by
Polyanskiy, Poor, and \verdu~\cite{polyanskiy2010channel} derives
finite-blocklength bounds directly from the channel statistics; the
error-exponent asymptotics and the fixed-error-probability (dispersion)
refinement are then recovered as
limits~\cite{strassen1962asymptotische,tomamichel2013tight}.

Two structural templates dominate the one-shot literature.  The information-spectrum approach
of \verdu{} and Han~\cite{verdu1994general,koga2002information}
reads both achievability and converse off a \emph{single} functional --- the
cumulative distribution of the information density
$\idens{X}{Y}=\log W(Y\mid X)/Q_Y(Y)$ --- at the price of a symmetric
$e^{-\gamma}$ slack.  The PPV approach~\cite{polyanskiy2010channel} instead
sharpens each direction with a functional of its own: a random-coding integral
(the RCU bound) on the achievability side, and a Neyman--Pearson
$\BETA{\alpha}{\cdot}{\cdot}$-functional (the meta-converse) on the converse
side.  The first template is structurally comparable but, specialized to concrete
finite-blocklength bounds, typically pays a common slack term; the second is
sharp but expresses the two directions through unrelated quantities.

Both templates are, in their canonical form, built around the
\emph{matched} maximum-likelihood metric: the information density $\idens{X}{Y}$
\emph{is} the log-likelihood ratio.  This ties them to a single decision rule
and leaves two questions open.  First, what governs achievability and converse
under an \emph{arbitrary} decoding metric --- a mismatched rule, a structured
decoder (matched filter, MMSE, zero-forcing), or a universal decoder whose
decision statistic is not a likelihood ratio?  Second, the meta-converse is
sharp only after optimizing over an input prior (and an auxiliary output law);
when is that prior-optimized value actually \emph{computable}, rather than the
output of a direct numerical minimax search?  This paper answers both through a
single primitive.

\subsection{The Pairwise Error Probability as a Primitive}

Fix a decoding metric $m:\cX\times\cY\to\mR$, an input prior $Q_X\in\cP(\cX)$,
and an output $y\in\cY$.  The \emph{pairwise error probability} (PEP) is
\[
  \PEP{x}{y}\;\triangleq\;\PRs{\bar X\sim Q_X}{m(\bar X,y)\succ m(x,y)},
\]
where $\succ$ is a strict comparison with a uniform (dithered) tie-breaker
made precise in \cref{sec:pep}: $\PEP{x}{y}$ is the probability
that a single random competing codeword outranks the transmitted one at output
$y$.  Only the relative order of scores enters this definition: any monotone
transformation of $m$ induces the same decoder, so the metric's numerical values
carry no meaning of their own.  The PEP is the canonical representative of this
order-equivalence class, and its value is itself an operational probability ---
the probability of being outranked --- rather than an arbitrary score.

The construction relies on one property of this object.  Under the
dithered tie rule, fix $y$ and draw $X\sim Q_X$ from the same prior as the
competitors: then $\PEP{X}{y}$ is \emph{exactly} uniform on $[0,1]$, regardless
of $y$, the channel, or the metric.  This uniformity is conditional on $y$ and
lifts to the product law $(X,Y)\sim Q_X\times Q_Y$ for any output measure
$Q_Y$ --- the form the meta-converse consumes.  Under the channel-induced joint
law $(X,Y)\sim Q_X\cdot\W$, by contrast, the conditional $X\mid Y$ is the Bayes
posterior, so $\PEP{X}{Y}$ is in general \emph{not} uniform: its law is the
\emph{error spectrum}
\[
  \Fspec{Q_X}{z}\;\triangleq\;\PR{-\log\PEP{X}{Y}\le z},
\]
the CDF that carries the dependence on $Q_X$ and the channel.  This
product-law versus channel-joint distinction is made precise in \cref{sec:pep}.

The PEP concerns a \emph{single} competitor; an actual decoding error involves
all $M-1$ of them.  For a random codebook the competitors are i.i.d., so,
conditioned on the transmitted pair, the codebook error probability is
\emph{exactly} $1-(1-\PEP{X}{Y})^{M-1}$ --- a deterministic function of the
pairwise quantity.  This \emph{lift} from one competitor to $M-1$ is the
mechanism of the paper: every random-coding bound below is an expectation
$\E{f(\PEP{X}{Y})}$ of the lift or of a relaxation of it, and we call the
function $f$ the \emph{kernel} of the bound.  The spectrum therefore determines
every kernel-based bound considered below, and both directions of the coding
problem are read from it at the \emph{same} point $z=R$ (with $R=\log M$).
Achievability integrates the spectrum against the random-coding kernel over a
vicinity of $z=R$.  The converse enters structurally rather than through a
separate argument: for a \emph{fixed} codebook the maximum-metric decoder errs
exactly when the transmitted codeword is outranked --- a threshold event on its
PEP --- so the fixed-code error equals $\Fspec{Q_X^\cC}{R}$ \emph{exactly}, for
\emph{any} metric, where $Q_X^\cC$ is the empirical input distribution of the
code $\cC$.  The two directions thus differ only by the kernel (smooth average
versus sharp threshold) and by the prior at which the spectrum is read.  What
matched ML adds is not the converse identity but \emph{structure}: under
matched ML the variational identities of \cref{sec:pep} represent the spectrum
itself, jointly convex in the testing level and the prior, and prior
optimization becomes tractable.

This is the reason to base the theory on the PEP rather than on the
information density $\idens{X}{Y}$: the error spectrum provides a common
representation for both directions, and, being built from the
\emph{rank} of the metric rather than from a likelihood ratio, it is defined
for any decoding rule.

\subsection{Contributions}

The contributions of the paper are the following.

\begin{enumerate}
\item \textbf{The dithered-PEP representation of both coding directions
(\cref{sec:pep,sec:channel}).}  Exact uniformity of the PEP under the
competitor prior (\cref{thm:pep-uniform}; atomlessness of the spectrum in
\cref{lem:atomless}) makes the spectrum the common object of both directions:
the exact \emph{ensemble-average} random-coding probability and its $\RCUp$
refinement are kernel averages of the spectrum
(\cref{thm:exact-rc,thm:rcu-plus}), and the error probability of any fixed
code equals the spectrum at the code's empirical prior, exactly and for any
metric (\cref{thm:cc-converse}).

\item \textbf{The variational identities (\cref{sec:pep}).}  For an arbitrary
metric, the metric-weighted tail functional of the PEP equals a Neyman--Pearson
$\BETA{\alpha}{\cdot}{\cdot}$-functional maximized over auxiliary output laws
(\cref{thm:meta-converse}) and, dually, a reverse-channel program under an
essential-supremum cap (\cref{thm:reverse-channel}); under matched ML the
weighted functional becomes the spectrum itself, which is then jointly convex
in the testing level and the prior (\cref{thm:joint-convex}), adapting the
saddle-point structure of the $\beta$-functional~\cite{polyanskiy2013saddle}.
The identification of the weighted functional with the spectrum, and the
resulting convexity of the spectrum, are restricted to matched ML.

\item \textbf{The minimax-prior LP (\cref{sec:prioropt}).}  Joint convexity
and the reverse-channel identity turn the prior-optimized minimax
meta-converse into a single \emph{finite-dimensional linear program}, jointly
in the input prior and an auxiliary reverse channel (\cref{thm:lp}), which
coincides with the finite-blocklength instance of Matthews' nonsignaling
converse~\cite{matthews2012linear}, an identification made exact through the
LP dual (\cref{rem:lp-dual}).  For memoryless channels with fixed alphabets, a
type reduction (\cref{prop:lp-type}) makes the program polynomial in the
blocklength, with degree growing with the alphabet size.

\item \textbf{The concave achievability program (\cref{sec:ach-prioropt}).}
Prior optimization of the random-coding bound --- which couples all
thresholds, not the converse's single one --- is a concave program over the
input simplex with an explicit gradient (\cref{prop:ach-concave}), optimized
by a first-order method on the simplex with the convergence guarantee of
\cref{prop:ach-march}; the same solver applies to the random-coding kernels
considered here.
\end{enumerate}

Consequences worked out along the way include recoveries of several classical
one-shot bounds --- Feinstein's threshold bound~\cite{feinstein1955new}, the
Han--\verdu{} information-spectrum bounds~\cite{verdu1994general,koga2002information},
the PPV RCU and minimax meta-converse~\cite{polyanskiy2010channel},
V\'azquez-Vilar's fixed-code exactness~\cite{vazquez2016bayesian}, and
Matthews' linear-programming converse~\cite{matthews2012linear}, each by a
short specialization at the point of recovery in
\cref{sec:channel,sec:prioropt} --- a rate-gap theorem: under
log-concavity of the spectrum, the achievability--converse rate gap at a fixed
error level and fixed prior is at most $-\log(1-\varsigma)/\varsigma$ in the
spectrum log-slope $\varsigma$, which approaches one nat in the small-slope
limit (\cref{thm:rategap}) --- and a local pairwise estimate providing the
$n^{-1/2}$ factor underlying the $\tfrac12\log n$ third-order improvement for
non-singular memoryless channels (\cref{lem:pep-local,rem:gauss-regime}).
The framework is not channel-specific: for Gallager-symmetric channels the
uniform input is optimal in both directions at every blocklength
(\cref{prop:symmetric}).  Settings that share the same spectrum but are not
pursued here --- joint source--channel and lossy source coding, refinements of
the random-coding bound, and problems outside the matched-ML convexity regime
--- are noted in \cref{sec:conc}.

\subsection{Relation to Prior Work}

The structural template --- one functional read in both directions --- is that
of the information-spectrum method of \verdu{} and
Han~\cite{verdu1994general,koga2002information}; under matched ML the PEP
meta-converse reduces to the same Neyman--Pearson test, and the PEP
construction extends the template to arbitrary metrics, where the
information-density route does not directly apply.  The mismatched-decoding
literature --- see the monograph of Scarlett, Guill\'en i F\`abregas,
Somekh-Baruch, and Martinez~\cite{scarlett2020mismatched} --- studies the rates
and error exponents attainable with a \emph{fixed} mismatched metric, and is
complementary to the present paper, which gives exact fixed-code identities for
the decoder actually used rather than optimal rates for a class of metrics.
The sharp finite-blocklength viewpoint, the RCU bound, and the meta-converse
are due to Polyanskiy, Poor, and \verdu~\cite{polyanskiy2010channel};
hypothesis-testing converses go back to Blahut~\cite{BlahutBook} and were
developed by Hayashi and Nagaoka~\cite{hayashi2003general} and Wang and
Renner~\cite{wang2012oneshot}; fixed-code exactness of the meta-converse under
matched ML is due to V\'azquez-Vilar et al.~\cite{vazquez2016bayesian}; the
convex--concave saddle structure of the $\beta$-functional is
Polyanskiy's~\cite{polyanskiy2013saddle}, and the linear-programming view of
channel-coding converses is Matthews'~\cite{matthews2012linear}.  Feinstein's
maximal-error threshold lemma~\cite{feinstein1955new} (maximal and average
error differing by at most a factor of two~\cite{shannon1957certain}) and the
exact dithered random-coding form~\cite{HaimKE18} complete the achievability
picture.  The Poisson-matching framework of Li and
Anantharam~\cite{li2021unified} likewise derives one-shot achievability bounds
from a single lemma, recovering second-order asymptotics across a range of
settings; it addresses the achievability direction only, whereas the spectrum
here carries both directions, and \cref{lem:pep-local} provides the local
$n^{-1/2}$ estimate underlying the $\tfrac12\log n$ third-order improvement
(\cref{rem:gauss-regime}).  What is new here is not any
single one of these ingredients but
their derivation from one spectrum, together with the formulation of the
prior-optimized minimax meta-converse as a finite-dimensional linear program
--- polynomial in the blocklength for memoryless channels with fixed alphabets
(\cref{prop:lp-type}).

The paper also consolidates a line of the authors' earlier work, and the
overlap should be delimited.  Achievability and converse bounds over a general
channel and a general decoding metric --- the comparable-pair template that
\cref{sec:channel} sharpens --- were developed in~\cite{elkayam2015achievable}
and its extended version~\cite{Elkayam14GeneralBounds}, which also contain a
precursor of the slope identity~\eqref{eq:slope-identity}; the variational
treatment of the $\beta$-functional that the identities of \cref{sec:pep} build
on was given in~\cite{elkayam2016variational}; and the computation of the
minimax converse by alternating saddle-point optimization was developed
in~\cite{elkayam2017calculation}.  The present paper builds on those works but
is not a union of them: the uniform-dither canonicalization of the PEP with its
exact probability-integral-transform identity (\cref{thm:pep-uniform}), the
exact fixed-code converse identity (\cref{thm:cc-converse}), the sharpened
slope identity with the strict bound $-1<\dot E\le0$ and the rate-gap theorem
(\cref{thm:rategap}), the single-LP formulation of the prior-optimized minimax
converse with its Matthews identification (\cref{thm:lp}), and the concave
achievability program of \cref{sec:ach-prioropt} appear in none of them.

\subsection{Organization}

\cref{sec:pep} defines the dithered PEP and establishes uniformity
(\cref{thm:pep-uniform}), the meta-converse and reverse-channel identities
(\cref{thm:meta-converse,thm:reverse-channel}), and joint convexity
(\cref{thm:joint-convex}).  \cref{sec:channel} applies them to point-to-point
channel coding --- achievability, the fixed-code converse, and the rate gap.
\cref{sec:prioropt,sec:ach-prioropt} develop the converse LP and its concave
achievability companion.  \cref{sec:num} reports numerical evaluations, and
\cref{sec:conc} concludes; throughout, logarithms are natural unless stated
otherwise, $\cP(\cX)$ denotes the probability simplex on the alphabet $\cX$,
and $\W$ denotes the channel.

\section{The Pairwise Error Probability and the Error Spectrum}\label{sec:pep}

This section develops the central primitive of the framework: a randomized (dithered) pairwise error probability and its associated \emph{error spectrum}. From a single uniformity property we derive three variational identities---a Neyman--Pearson (meta-converse) form, a reverse-channel form, and a joint-convexity property. The first two represent metric-weighted tail functionals of the PEP and hold for an arbitrary decoding metric; under the matched likelihood-ratio metric the represented functional becomes the spectrum itself, which the third property then makes jointly convex in the testing level and the prior. All subsequent sections use these identities as stated.

\subsection{The dithered pairwise error probability}

Throughout, the alphabets $\cX$ and $\cY$ are standard Borel spaces, $\cP(\cX)$ denotes the Borel probability measures on $\cX$, and every score is jointly measurable; by Fubini's theorem this makes the derived quantities below ($G_{x,y}$, $H_{x,y}$, and the spectrum) measurable in their arguments. \Cref{sec:prioropt,sec:ach-prioropt} specialize to finite alphabets.

Fix a prior $Q_X\in\cP(\cX)$ and a (jointly) measurable score (metric) $m:\cX\times\cY\to[0,\infty]$. For a fixed observation $y\in\cY$, the relevant reliability question is how likely a randomly drawn competitor $X\sim Q_X$ is to outrank a given candidate $x$ under the score $m(\cdot,y)$. Only the relative order of scores matters --- any monotone transformation of $m$ induces the same comparisons --- and ties must be resolved without bias. To this end we augment each score with an auxiliary uniform variable (a \emph{dither}) $U\sim\uU$ and order score--dither pairs lexicographically,
\begin{equation}\label{eq:lex-order}
  (m_1,u_1)\succ(m_2,u_2)\iff
  \BRAb{m_1>m_2}\ \text{or}\ \BRAb{m_1=m_2\ \text{and}\ u_1<u_2},
\end{equation}
the reversed $u$-coordinate ensuring that, among tied scores, the competitor wins the tie-break with probability $\PR{U<u}=u$.

\begin{definition}[Dithered PEP]\label{def:pep}
For $x\in\cX$, $y\in\cY$, $u\in[0,1]$, set
\begin{align*}
  G_{x,y}&\triangleq Q_X\BRA{\BRAs{\bar x:m(\bar x,y)>m(x,y)}},\\
  H_{x,y}&\triangleq Q_X\BRA{\BRAs{\bar x:m(\bar x,y)=m(x,y)}},
\end{align*}
and define the \emph{pairwise error probability} (PEP) of the triple $(x,y,u)$ by
\begin{equation}\label{eq:pep-def}
  \PEPU{x}{y}{u}\triangleq G_{x,y}+u\,H_{x,y}.
\end{equation}
The complementary \emph{correct probability} is $\PECU{x}{y}{u}\triangleq 1-\PEPU{x}{y}{u}$. Substituting a fresh $U\sim\uU$ gives the \emph{randomized PEP} $\PEP{x}{y}\triangleq\PEPU{x}{y}{U}$, a $[0,1]$-valued random variable.
\end{definition}

The quantity $\PEPU{x}{y}{u}$ is exactly the probability that a random pair $(X,U)\sim Q_X\times\uU$ outranks $(x,u)$ under~\eqref{eq:lex-order}: the strictly-larger-score mass $G_{x,y}$ plus the tie mass $H_{x,y}$ weighted by the chance $u$ that the competitor wins the tie. Thus the PEP maps every candidate onto a common $[0,1]$ scale of risk, and lexicographic comparison reduces to numerical comparison of PEP values. In this sense the PEP is the canonical representative of the metric's order-equivalence class: any monotone transformation of $m(\cdot,y)$ leaves it unchanged, and its value is itself an operational probability --- the probability of being outranked --- rather than an arbitrary score.

The randomized tie-break provides a uniform treatment of ties: it affects only the tie set $\{x':m(x',y)=m(x,y)\}$, and its purpose is to expose the exact probability-integral-transform identity of the next subsection, replacing the case-by-case tie conventions (fair-split for the exact analysis, ties-as-errors for the union bound) by a single rule that is tight in both directions. For continuous-valued metrics ties typically have probability zero, and the randomized rule then coincides almost surely with any deterministic tie-breaking. For discrete channels the results apply as stated to the randomized tie rule; \cref{rem:converse-decoder-scope} delimits which of them extend to deterministic tie-breaking.

\subsection{Uniformity and the error spectrum}

The defining property of the construction is that, when the candidate is itself drawn from the prior, the PEP is exactly uniform at every fixed output---a randomized probability-integral transform that holds irrespective of $y$, the channel, or the metric.

\begin{lemma}[PEP properties]\label{thm:pep-uniform}
Let $(X,U)\sim Q_X\times\uU$. Then:
\begin{enumerate}
\item\emph{(Lexicographic equivalence.)} For every fixed $(x,y,u)$,
\[
  \PRs{X,U}{(m(X,y),U)\succ(m(x,y),u)}=\PRs{X,U}{\PEP{X}{y}<\PEPU{x}{y}{u}}=\PEPU{x}{y}{u}.
\]
\item\emph{(Order preservation.)} If $m(x_1,y)>m(x_2,y)$ then $\PEPU{x_1}{y}{u_1}\le\PEPU{x_2}{y}{u_2}$ for all $u_1,u_2$.
\item\emph{(Uniformity.)} For every fixed $y$,
\begin{equation}\label{eq:pep-uniform}
  \PEP{X}{y}=\PEPU{X}{y}{U}\sim\uU.
\end{equation}
\end{enumerate}
\end{lemma}

\begin{IEEEproof}
Properties (i)--(ii) are direct from~\eqref{eq:pep-def} and the definition of $\succ$. For (iii), with $T\triangleq m(X,y)$ of CDF $F_T$ and atom function $P_T$, the identity
\[
\PEPU{X}{y}{U}=1-\BRA{F_T(T^-)+(1-U)P_T(T)}
\]
reduces uniformity to the randomized probability-integral transform: for $T$ with $U\sim\uU$ independent, $F_T(T^-)+U\,P_T(T)\sim\uU$, the dither filling each atom's height interval to restore uniformity; see Appendix~\ref{app:pep-uniform} for the full argument.
\end{IEEEproof}

Uniformity at fixed $y$ lifts immediately to the product law: under $(X,Y)\sim Q_X\times Q_Y$ for \emph{any} output measure $Q_Y\in\cP(\cY)$, the randomized PEP satisfies $\PEP{X}{Y}\sim\uU$, since~\eqref{eq:pep-uniform} holds for every $y$ separately. This product-law uniformity is the precise form consumed by the meta-converse identity below.

The situation is qualitatively different under the \emph{channel-induced joint}. When $(X,Y)\sim Q_X\cdot\W$, the conditional $X\mid Y$ is the Bayes posterior rather than the prior, so the cancellation behind~\eqref{eq:pep-uniform} no longer occurs and $\PEP{X}{Y}$ is in general \emph{not} uniform. Its law is the nontrivial object that the achievability analysis optimizes.

\begin{definition}[Error spectrum]\label{def:spectrum}
For $X\sim Q_X$, $Y\sim\W(\cdot\mid X)$, the \emph{(PEP) error spectrum} is the cumulative distribution
\begin{equation}\label{eq:spectrum-def}
  \Fspec{Q_X}{z}\triangleq\PR{-\log\PEP{X}{Y}\le z},\qquad z\ge 0.
\end{equation}
\end{definition}

The spectrum depends on $Q_X$ and $\W$, and the bounds developed below can be expressed in terms of it; the joint source--channel spectrum is the identical construction over an enlarged candidate space, and the distortion spectrum is the same lift run in the value domain (the rank transform of the distortion) --- neither introduces a new informational quantity, though we do not pursue those settings here.

\subsection{Atomlessness under mismatch}

Several arguments---in particular the random-coding-to-converse passage of Section~\ref{sec:channel}---require that the PEP have no atoms, so that strict threshold inequalities pass to equalities almost surely. Under matched maximum-likelihood decoding the dither already removes atoms; under mismatch we invoke absolute continuity.

\begin{lemma}[Atomlessness]\label{lem:atomless}
Let $Q_X^{\mathrm{pep}}\in\cP(\cX)$ define the PEP and let $Q_X^{\mathrm{prob}}\in\cP(\cX)$ with $Q_X^{\mathrm{prob}}\ll Q_X^{\mathrm{pep}}$. Let $\W$ be a channel whose output marginal $P_Y$ (induced by $Q_X^{\mathrm{prob}}\cdot\W$) satisfies $Q_X^{\mathrm{prob}}\cdot\W\ll Q_X^{\mathrm{prob}}\times P_Y$ --- automatic for discrete alphabets, and whenever $\W(\cdot\mid x)$ admits a density with respect to $P_Y$. For $(X,Y)\sim Q_X^{\mathrm{prob}}\cdot\W$,
\[
  \PR{\PEP{X}{Y}=w}=0\qquad\text{for every }w\in[0,1];
\]
consequently $Z_e\triangleq-\log\PEP{X}{Y}$ has a continuous CDF.
\end{lemma}

\begin{IEEEproof}
The two absolute-continuity hypotheses chain: $Q_X^{\mathrm{prob}}\cdot\W\ll Q_X^{\mathrm{prob}}\times P_Y\ll Q_X^{\mathrm{pep}}\times P_Y$, the first by hypothesis and the second from $Q_X^{\mathrm{prob}}\ll Q_X^{\mathrm{pep}}$. Under $Q_X^{\mathrm{pep}}\times P_Y$ the PEP is uniform by Lemma~\ref{thm:pep-uniform}, hence atomless; absolute continuity transports null sets.
\end{IEEEproof}

\subsection{The meta-converse identity}

Product-law uniformity identifies the PEP threshold event with the acceptance region of a Neyman--Pearson test, turning the metric-weighted tail functional of the PEP into a $\beta$-functional supremized over an auxiliary output distribution; in the matched case~\eqref{eq:meta-channel} the represented functional is the spectrum itself. For finite nonnegative measures $P,Q$ on $\cX\times\cY$,
\[
  \BETA{\alpha}{P}{Q}\triangleq\min_{0\le T\le 1}\ \Es{Q}{T}\quad\text{s.t.}\quad\Es{P}{T}\ge\alpha,
\]
the Neyman--Pearson functional, with optimum attained by a likelihood-ratio threshold test. Write $((Q_X\times P_Y)\cdot m)(x,y)\triangleq Q_X(x)P_Y(y)m(x,y)$ for the metric-tilted measure --- the reference distribution that reweights each pair $(x,y)$ by the metric value $m(x,y)$, playing the role that the output distribution plays in the PPV binary hypothesis test.

\begin{theorem}[Meta-converse identity]\label{thm:meta-converse}
Let $Q_X\in\cP(\cX)$, $P_Y\in\cP(\cY)$ and $m:\cX\times\cY\to[0,\infty]$. Assume \textup{(A1)} $Q_X\{x:m(x,y)<\infty\}=1$ for every $y$ in the support of $P_Y$, and \textup{(A2)} $\Es{Q_X\times P_Y}{m(X,Y)}<\infty$. With $(X,Y)\sim Q_X\times P_Y$ and $Z_e\triangleq-\log\PEP{X}{Y}$, for every $R>0$:
\begin{equation}
  \Es{Q_X\times P_Y}{m(X,Y)\,\Ind{Z_e\le R}}
    =\max_{Q_Y\in\cP(\cY)}\BETA{1-e^{-R}}{Q_X\times Q_Y}{(Q_X\times P_Y)\cdot m}.\label{eq:meta-converse-e}
\end{equation}
\end{theorem}

\begin{IEEEproof}
The test $T_e(x,y)=\Ind{\PEP{x}{y}\ge e^{-R}}$ has $Q_X$-mass $1-e^{-R}$ at each fixed $y$ by uniformity, hence is feasible at every $Q_Y$, giving ``$\le$''. For ``$\ge$'', a randomized output quantile $\tau_y,\theta_y$ solving $Q_X(m(\cdot,y)<\tau_y)+\theta_y\,Q_X(m(\cdot,y)=\tau_y)=1-e^{-R}$ defines $Q_Y^\star(y)\propto P_Y(y)\tau_y$ (finite by (A2) via Markov's inequality); under $Q_Y^\star$ the likelihood ratio is $\lambda^\star m(x,y)/\tau_y$, so the Neyman--Pearson test coincides exactly with $T_e$, achieving equality. See Appendix~\ref{app:meta-converse} for the full argument. A companion identity for the complementary (correct-decoding) tail $Z_c\triangleq-\log\PEC{X}{Y}$ holds by the same argument; it is not used in this paper.
\end{IEEEproof}

Assumptions \textup{(A1)}--\textup{(A2)} are mild and hold automatically in the standard discrete-memoryless setting under matched maximum-likelihood decoding, where the metric is everywhere finite with finite mean. Specialized to channel coding with the maximum-likelihood metric $m(x,y)=\W(y\mid x)/P_Y(y)$, the identity reads, for $(X,Y)\sim Q_X\cdot\W$,
\begin{equation}\label{eq:meta-channel}
  \Fspec{Q_X}{z}=\PRs{Q_X\cdot\W}{Z_e\le z}
    =\max_{Q_Y\in\cP(\cY)}\BETA{1-e^{-z}}{Q_X\times Q_Y}{Q_X\cdot\W},
\end{equation}
exhibiting the spectrum as an output-optimized binary hypothesis test. The same identity holds over an enlarged candidate space (for instance source--codeword pairs, for joint source--channel coding), an extension we do not pursue here.

\subsection{The reverse-channel identity}

The meta-converse form tests on the joint $(X,Y)$ space; the next form re-parametrizes the same quantity through a \emph{reverse channel} $W_{X\mid Y}$ subject to an essential-supremum divergence cap. This is the representation that the linear-programming formulation of Section~\ref{sec:prioropt} consumes. For a reverse channel with $W_{X\mid Y}(\cdot\mid y)\ll Q_X$ for $Q_Y$-a.e.\ $y$,
\[
  D_\infty\BRA{W_{X\mid Y}\,\|\,Q_X\mid Q_Y}\triangleq\esssup_{y\sim Q_Y}\ \sup_{x}\ \log\frac{W_{X\mid Y}(x\mid y)}{Q_X(x)},
\]
so the cap $D_\infty\le z$ is equivalent to the pointwise constraint $W_{X\mid Y}(x\mid y)\le e^{z}Q_X(x)$.

\begin{theorem}[Reverse-channel identity]\label{thm:reverse-channel}
Let $m:\cX\times\cY\to[0,\infty]$ satisfy the analogues of \textup{(A1)--(A2)} with $P_Y$ replaced by $Q_Y$. For every $z>0$,
\begin{align}
  \Es{Q_X\times Q_Y}{m(X,Y)\,\Ind{\PEC{X}{Y}\le e^{-z}}}
    &=\inf_{W:\,D_\infty(W\|Q_X\mid Q_Y)\le z}e^{-z}\,\Es{W\cdot Q_Y}{m(X,Y)},\label{eq:rc-inf}\\
  \Es{Q_X\times Q_Y}{m(X,Y)\,\Ind{\PEP{X}{Y}\le e^{-z}}}
    &=\sup_{W:\,D_\infty(W\|Q_X\mid Q_Y)\le z}e^{-z}\,\Es{W\cdot Q_Y}{m(X,Y)}.\label{eq:rc-sup}
\end{align}
\end{theorem}

\begin{IEEEproof}
By uniformity the event $\{\PEP{X}{Y}\le e^{-z}\}$ has $Q_X$-mass exactly $e^{-z}$ at each $y$, so $W^\star(x\mid y)\triangleq e^{z}Q_X(x)\PR{\PEP{x}{y}\le e^{-z}}$ is a valid reverse channel saturating the cap and supported on this event. Since the PEP is monotone (decreasing) in the metric rank, this event is a superlevel set of $m(\cdot,y)$; a bathtub (rearrangement) exchange argument---``fill from the top''---shows $W^\star$ maximizes the linear objective under the box constraint, giving~\eqref{eq:rc-sup}. The inf form~\eqref{eq:rc-inf} is the sign-reversed ``fill-from-the-bottom'' construction on the sublevel set $\{\PEC{X}{Y}\le e^{-z}\}$; see Appendix~\ref{app:reverse-channel} for the full argument.
\end{IEEEproof}

The two forms are dual: the correct-decoding event picks out \emph{sublevel} sets of $m(\cdot,y)$ (lower tail, $\inf$), the error event \emph{superlevel} sets (upper tail, $\sup$). Each identity is pointwise in the threshold, hence extends by linearity to any nonnegative combination of thresholds, the reverse channels at distinct thresholds remaining independent optimization variables.

\subsection{Joint convexity in the prior}

The final structural property makes prior optimization tractable: the output-optimized $\beta$-functional is jointly convex in the level and the prior.

\begin{theorem}[Joint convexity]\label{thm:joint-convex}
Fix $P_Y\in\cP(\cY)$ and a nonnegative kernel $m:\cX\times\cY\to[0,\infty)$, and set $g(x)\triangleq\Es{P_Y}{m(x,Y)}$. Then
\[
  F(\alpha,Q_X)=\max_{Q_Y\in\cP(\cY)}\BETA{\alpha}{Q_X\times Q_Y}{(Q_X\times P_Y)\cdot m}
\]
is jointly convex in $(\alpha,Q_X)$ on $[0,1]\times\cD$, where $\cD\triangleq\BRAs{Q_X\in\cP(\cX):g(x)<\infty\ \text{for }Q_X\text{-a.e.\ }x}$. The finiteness hypothesis is imposed on each prior entering a convex combination --- $g(x)<\infty$ for $Q_X^{(j)}$-a.e.\ $x$, for each $j$ --- and passes to the mixture $\sum_j\lambda_j Q_X^{(j)}$, so $\cD$ is convex.
\end{theorem}

\begin{IEEEproof}
At each fixed $Q_Y$, joint convexity of $(\alpha,Q_X)\mapsto\beta_\alpha$ follows by mixing optimal tests, exactly as in the saddle-point convexity argument of~\cite{polyanskiy2013saddle}; the finiteness hypothesis $Q_X^{(j)}\in\cD$ keeps every term of the mixture finite, and the pointwise maximum over $Q_Y$ preserves joint convexity. Details omitted.
\end{IEEEproof}

Geometrically, $\beta_\alpha$ is the value of a linear program, convex in the constraint level $\alpha$ and concave in the cost, into which $Q_X$ enters linearly through the mixture distribution; it traces a convex Pareto frontier in the (type-I, type-II) plane, and the supremum over $Q_Y$ inherits joint convexity. The hypothesis is automatic in the channel-coding setting, where $g(x)$ reduces to $1$, and in the joint source--channel setting, where it reduces to a source likelihood ratio. This convexity, together with the threshold-wise separability of the reverse-channel identity, is precisely what permits prior optimization to be carried out as a convex---indeed linear---program in Section~\ref{sec:prioropt}.

\section{Point-to-Point Channel Coding}\label{sec:channel}

This section applies the PEP primitive of Section~\ref{sec:pep} to point-to-point channel coding. Computing the exact optimal one-shot error probability is believed to be NP-hard~\cite{costa2010one}; we bracket it by a comparable achievability/converse pair, both read off the same error spectrum $\Fspec{Q_X}{z}$.

Throughout, fix a channel $\W:\cX\to\cY$ and a decoding metric $m:\cX\times\cY\to[0,\infty)$, and transmit a message $I\sim\Unif{[M]}$ as $X=\varphi(I)$ through $\W$. A codebook $\cC=\BRAs{x_1,\dots,x_M}$ induces the encoder $\varphi(i)=x_i$; the decoder selects $\cD(y)=\argmax_i m(x_i,y)$, with ties resolved by the dither rule of Section~\ref{sec:pep}. The block error probability of a fixed code is $P_e(\cC)\triangleq\PR{\cD(Y)\neq I}$, and the optimal probability at rate $R\ge0$ is $P_e(R)\triangleq\inf_{\abs{\cC}=\lceil e^R\rceil}P_e(\cC)$.

A \emph{random codebook} $\bC=\BRAs{X_1,\dots,X_M}$ draws each codeword i.i.d.\ from a prior $Q_X\in\cP(\cX)$; its average error probability is $\bar P_e(R;Q_X)\triangleq\Es{\bC}{P_e(\bC)}$, and the probabilistic method gives $P_e(R)\le\bar P_e(R;Q_X)$. The design problem thus reduces to optimizing over the prior $Q_X$. The dither-based tie rule, recalled next, makes this random-coding analysis \emph{exact}.

\subsection{Dithered tie-breaking and the PEP}

Decoding uses the construction of Section~\ref{sec:pep} verbatim: each codeword carries an i.i.d.\ dither $U_i\sim\uU$, score--dither pairs are compared by the lexicographic order~\eqref{eq:lex-order}, and the decoder is $\cD(y)=\argmax_{1\le i\le M}(m(X_i,y),U_i)$ with respect to $\succ$ --- a rule that is almost surely decisive. With the dithered PEP $\PEPU{x}{y}{u}=G_{x,y}+u\,H_{x,y}$ of Definition~\ref{def:pep} and its randomized form $\PEP{x}{y}=\PEPU{x}{y}{U}$, two facts drive the analysis: decoding by pairs is probabilistically equivalent to comparing PEPs, and a random competitor's PEP is uniform on $[0,1]$ (Lemma~\ref{thm:pep-uniform}):
\begin{align}
\PR{(m(X,y),U)\succ(m(x,y),u)}
&=\PR{\PEP{X}{y}<\PEPU{x}{y}{u}}\label{eq:ch-lex}\\
&=\PEPU{x}{y}{u}.\label{eq:ch-unif}
\end{align}

\subsection{Exact random-coding performance}

The first result is an exact evaluation of the ensemble-average error probability, equivalent to the fair-split expression \cite[Thm.~15]{polyanskiy2010channel} but streamlined through the PEP. ``Exact'' is meant at the level of the ensemble average and finite codebook size $M$: no asymptotic limit and no relaxation of the codeword-pair counting.

\begin{theorem}[Exact ensemble-average random coding]\label{thm:exact-rc}
For $M=\lceil e^R\rceil$ codewords drawn i.i.d.\ from $Q_X$,
\begin{equation}\label{eq:exact-rc}
\bar P_e(R;Q_X)=\Es{X,Y,U}{1-(1-\PEPU{X}{Y}{U})^{M-1}},
\end{equation}
where $(X,Y)\sim Q_X\cdot\W$ and $U\sim\uU$.
\end{theorem}
\begin{IEEEproof}
By symmetry assume the first codeword is sent. Conditioned on $(X_1,Y,U_1)$, the $M-1$ competitors are independent, so the correct-decoding probability factors as
\[
\prod_{i=2}^{M}\PR{(m(X_1,Y),U_1)\succ(m(X_i,Y),U_i)}.
\]
By \eqref{eq:ch-lex}--\eqref{eq:ch-unif} each factor equals $1-\PEPU{X_1}{Y}{U_1}$, giving $(1-\PEPU{X_1}{Y}{U_1})^{M-1}$; taking expectations yields \eqref{eq:exact-rc}.
\end{IEEEproof}

Using $\E{U^k}=1/(k+1)$, \eqref{eq:exact-rc} is algebraically equivalent to \cite[Eq.~55]{polyanskiy2010channel}, recovering the exact random-coding probability of Polyanskiy, Poor and \verdu, and matching the dithered expression of \cite{HaimKE18}.

\subsection{The \texorpdfstring{$\RCUp$}{RCU+} bound}

Clipping the union bound turns the exact expression into a tractable achievability bound.

\begin{theorem}[$\RCUp$ bound]\label{thm:rcu-plus}
With $M=\lceil e^R\rceil$,
\begin{equation}\label{eq:rcu-plus}
\bar P_e(R;Q_X)\le\tilde P_e(R;Q_X)\triangleq\Es{X,Y,U}{\min\BRA{1,e^R\,\PEPU{X}{Y}{U}}}.
\end{equation}
\end{theorem}
\begin{IEEEproof}
Immediate from $1-(1-x)^{M-1}\le\min(1,(M-1)x)$ applied to \eqref{eq:exact-rc}, followed by $M-1\le e^{R}$ for $M=\lceil e^R\rceil$, which bounds the clipped slope by $e^{R}$ and keeps the bound defined for all real $R$.
\end{IEEEproof}

The functional $\tilde P_e$, which we call $\RCUp$, is a refinement of the clipped union bound. It is tighter than the PPV $\mathrm{RCU}$ bound \cite[Thm.~16]{polyanskiy2010channel}, which follows from it by the relaxation $\PEP{X}{Y}\le G_{X,Y}+H_{X,Y}$. Moreover, by concavity of $w\mapsto\min\BRAs{1,w}$ and Jensen's inequality applied to the inner expectation over $U$,
\begin{equation}\label{eq:rcu-star}
\tilde P_e(R;Q_X)\le\Es{X,Y}{\min\BRA{1,e^R\BRA{G_{X,Y}+\tfrac12 H_{X,Y}}}},
\end{equation}
which is the $\mathrm{RCU}^*$ bound of \cite[Eq.~23]{HaimKE18}, provably tighter than both $\mathrm{RCU}$ and the dependence-testing ($\mathrm{DT}$) bound of~\cite{polyanskiy2010channel}.

Specialized to matched maximum-likelihood decoding, the bound contains the classical threshold bounds by two elementary steps. Splitting the clipped kernel at a shifted threshold gives, for every $\gamma>0$, $\tilde P_e\le\Fspec{Q_X}{R+\gamma}+e^{-\gamma}$; weakening the spectrum by the change-of-measure estimate $\PEP{x}{y}\le G_{x,y}+H_{x,y}\le e^{-\idens{x}{y}}$ (the information density taken with respect to the $Q_X$-induced output law) gives $\PR{\idens{X}{Y}\le R+\gamma}+e^{-\gamma}$ --- the Han--\verdu\ information-spectrum direct bound~\cite{verdu1994general,koga2002information}, whose maximal-error version, Feinstein's threshold bound~\cite{feinstein1955new}, follows by the standard average-to-maximal conversion~\cite{shannon1957certain}. Applying the same estimate inside the kernel instead gives $\E{\min\BRAs{1,e^{R}e^{-\idens{X}{Y}}}}=\E{e^{-[\idens{X}{Y}-R]^{+}}}$, the $\mathrm{DT}$ bound up to the constant in the threshold.

\subsection{Error spectrum and integral formula}

In analogy with information-spectrum methods, the primary informational quantity is the error spectrum of the negative log-PEP, $\Fspec{Q_X}{z}=\PR{-\log\PEP{X}{Y}\le z}$ with $(X,Y)\sim Q_X\cdot\W$ (\cref{def:spectrum}).

The map $z\mapsto\Fspec{Q_X}{z}$ is continuous, since the PEP distribution is atomless (Lemma~\ref{lem:atomless}; its absolute-continuity hypothesis --- automatic for discrete alphabets, and satisfied whenever the channel has output densities --- is assumed throughout this section). The $\RCUp$ bound is an expectation of a function of the PEP, so the spectrum carries all the information needed to compute it.

\begin{theorem}[Integral representation]\label{thm:rcu-integral}
\begin{equation}\label{eq:integral}
\tilde P_e(R;Q_X)=e^R\int_R^\infty e^{-z}\,\Fspec{Q_X}{z}\,dz.
\end{equation}
\end{theorem}
\begin{IEEEproof}
Writing $Z_e\triangleq-\log\PEP{X}{Y}$, whose CDF is the spectrum
$\Fspec{Q_X}{\cdot}$, we have
\[
\tilde P_e(R;Q_X)=\E{\min(1,e^{R-Z_e})}=\Fspec{Q_X}{R}+\int_R^\infty e^{R-z}\,d\Fspec{Q_X}{z};
\]
integration by parts gives \eqref{eq:integral}.
\end{IEEEproof}

Substituting $z=R+u$ in~\eqref{eq:integral} gives the probabilistic reading
\begin{equation}\label{eq:exp-window}
\tilde P_e(R;Q_X)=\int_0^\infty\Fspec{Q_X}{R+u}\,e^{-u}\,du=\E{\Fspec{Q_X}{R+T}},\qquad T\sim\mathrm{Exp}(1):
\end{equation}
the bound is the spectrum smoothed \emph{forward} through an exponential window of mean one nat. The window's unit decay rate is structural --- one extra nat of rate multiplies the competitor count $e^R$ by $e$, so spectrum mass one nat past the threshold is discounted by exactly a factor $e$ --- and it is fixed once and for all: it never sharpens. Which of its features the bound feels is decided by the spectrum: a steep spectrum is dominated by the window's tail, a flat one reads off its mean, a one-nat forward shift. This is the source of the rate gap quantified in \cref{sec:rategap}.

The behavior of $\tilde P_e$ is governed by the slope of its error exponent. Let $E(R)\triangleq-\log\tilde P_e(R;Q_X)$. Differentiating \eqref{eq:integral} yields the closed-form identity (Appendix~\ref{app:integral})
\begin{equation}\label{eq:slope-identity}
\tilde P_e(R;Q_X)=\frac{\Fspec{Q_X}{R}}{1+\dot E(R)},\qquad -1<\dot E(R)\le 0,
\end{equation}
valid for every $Q_X$ and every channel at any $R$ with $\Fspec{Q_X}{R}>0$. When $\dot E(R)\approx 0$ (steep exponent) the random-coding probability tracks the spectrum at $R$; as $\dot E(R)\to-1$ (straight-line regime) the integral tail dominates and the ratio $\tilde P_e/\Fspec{Q_X}{\cdot}$ is unbounded. The error-exponent envelope and the finite-blocklength normal (dispersion) approximation both follow from the integral formula~\eqref{eq:integral} by the same kernel argument; we do not develop these refinements here.

\subsection{Fixed-code converse via PEP}

For a fixed codebook the maximum-metric decoder errs exactly when the transmitted codeword is outranked --- a threshold event on its PEP; the same spectrum therefore governs a general converse, exact for \emph{any} code under an \emph{arbitrary} metric.

\begin{theorem}[Fixed-code converse equality]\label{thm:cc-converse}
For any codebook $\cC$ with $M$ codewords, let $Q_X$ be the uniform (empirical) distribution over $\cC$. Then, for the dithered maximum-metric decoder,
\begin{equation}\label{eq:cc-converse}
P_e(\cC)=\PR{\PEP{X}{Y}\ge 1/M}.
\end{equation}
The scope of the identity across other tie-breaking rules is delimited in Remark~\ref{rem:converse-decoder-scope} below.
\end{theorem}
\begin{IEEEproof}
Under the empirical prior $Q_X$, for each $i$ the value $\PEPU{x_i}{y}{U_i}=G_{x_i,y}+U_i H_{x_i,y}$ equals the (dithered) normalized rank of $x_i$ by score among the $M$ candidates. The transmitted codeword is mis-decoded iff its rank exceeds $1$, i.e.\ iff $\PEP{X}{Y}>1/M$; atomlessness (Lemma~\ref{lem:atomless}) gives $\PR{\PEP{X}{Y}=1/M}=0$, so the strict and non-strict events coincide. This parallels the rank-style converse of \cite[Lem.~1]{somekh2015general}, refined here to handle ties; see Appendix~\ref{app:integral} for the full argument.
\end{IEEEproof}

\begin{remark}[Which decoder ``exact'' refers to]\label{rem:converse-decoder-scope}
The error probability in \cref{thm:cc-converse} is that of the \emph{dithered} maximum-metric decoder --- a randomized tie rule. How the identity fares under other tie rules differs sharply between the matched and mismatched cases.

\emph{Matched case: the identity is tie-rule-independent and equals the optimal (MAP) error.} With $m(x,y)=\log\W(y\mid x)$ (or any metric inducing the same ordering of the codewords at each $y$) and uniform messages, the posterior of message $i$ given $Y=y$ is $\W(y\mid x_i)/\sum_j\W(y\mid x_j)$, strictly increasing in the metric value; the metric maximizers are exactly the posterior maximizers and \emph{all} carry the same, maximal, posterior. Every rule for choosing among them --- the dither, a fair split, or any deterministic rule --- attains the same conditional correct-decoding probability $\max_i\PR{I=i\mid Y=y}$, the MAP optimum. This equal-posterior argument is what legitimizes reading~\eqref{eq:mc-tight} below as the fixed-code exactness of~\cite{vazquez2016bayesian}, whose $P_e$ is the optimal decoder's.

\emph{Mismatched case with metric atoms: the identity is specific to the dithered rule.} Codewords tied in metric can carry unequal posteriors under the true channel, and a tie rule favoring the higher-posterior candidate achieves strictly smaller error; for the maximum-metric class with arbitrary tie-breaking, \eqref{eq:cc-converse} is then neither an equality nor a one-sided bound.

The discrepancy is confined to the tie event: all maximum-metric decoders coincide off the event that the top metric is attained by two or more codewords, so any tie rule's error differs from $\PR{\PEP{X}{Y}\ge1/M}$ by at most the probability of that top-rank tie --- zero whenever the metric ties carry no probability mass.
\end{remark}

The identity is exact, with no inequality, so the converse loss for a code reduces entirely to the gap between its empirical prior $Q_X$ and the prior-optimal $Q_X^\ast$. Combining achievability \eqref{eq:slope-identity} and converse \eqref{eq:cc-converse} brackets the optimal probability through the single spectrum:
\begin{equation}\label{eq:ach-conv-bracket}
\inf_{Q_X}\Fspec{Q_X}{R}\ \le\ P_e(R)\ \le\ \inf_{Q_X}\frac{\Fspec{Q_X}{R}}{1+\dot E(R)}.
\end{equation}
The prior plays structurally distinct roles on the two sides: the converse evaluates $\Fspec{Q_X}{R}$ at one threshold, whereas achievability optimizes a kernel-weighted integral \eqref{eq:integral} of the same spectrum. \Cref{sec:rategap} quantifies the width of the bracket~\eqref{eq:ach-conv-bracket} at a fixed prior; Section~\ref{sec:prioropt} develops the linear program that computes the converse optimum.

\subsection{Matched case and the meta-converse}

When the decoder is unconstrained, matched maximum-likelihood decoding is optimal, and the spectrum becomes the value of a binary hypothesis test, linking it to the meta-converse.

\begin{theorem}[Meta-converse identity, matched case]\label{thm:cc-matched}
In the matched case, with the likelihood-ratio metric $m(x,y)=\W(y|x)/P_Y(y)$ (order-equivalent in $x$ to the matched log-likelihood, hence inducing the same spectrum),
\begin{equation}\label{eq:meta-converse}
\Fspec{Q_X}{R}=\max_{Q_Y}\BETA{1-e^{-R}}{Q_X\times Q_Y}{Q_X\cdot\W}.
\end{equation}
\end{theorem}
\begin{IEEEproof}
This is the matched specialization of the PEP meta-converse identity (Theorem~\ref{thm:meta-converse}), whose assumptions \textup{(A1)}--\textup{(A2)} hold here as noted after that theorem: the metric is finite wherever $P_Y>0$, with $\Es{Q_X\times P_Y}{m(X,Y)}\le1$. The negative log-PEP threshold becomes the Neyman--Pearson level of the test between the product $Q_X\times Q_Y$ and the joint $Q_X\cdot\W$, optimized over the auxiliary output $Q_Y$; see Appendix~\ref{app:meta-converse}.
\end{IEEEproof}

Combining the fixed-code equality \eqref{eq:cc-converse} with \eqref{eq:meta-converse} recovers the fixed-code exactness of \cite[Thm.~1]{vazquez2016bayesian}: for the prior induced by any fixed code, the meta-converse coincides exactly with the true error probability,
\begin{equation}\label{eq:mc-tight}
P_e(\cC)=\max_{Q_Y}\BETA{1-e^{-R}}{Q_X\times Q_Y}{Q_X\cdot\W}.
\end{equation}
A code-independent bound requires minimizing over the input prior, yielding the PPV minimax meta-converse \cite[Thm.~27]{polyanskiy2010channel},
\begin{equation}\label{eq:minimax-mc}
\inf_{Q_X}\max_{Q_Y}\BETA{1-e^{-R}}{Q_X\times Q_Y}{Q_X\cdot\W}.
\end{equation}
Its tractability hinges on the convexity of $Q_X\mapsto\Fspec{Q_X}{z}$, which holds in the matched case for every $z$ (a consequence of the joint convexity of Theorem~\ref{thm:joint-convex}). The resulting convex program --- a finite-dimensional LP, polynomial in the blocklength for memoryless channels with fixed alphabets after a type reduction --- is the subject of Section~\ref{sec:prioropt}.

\subsection{Worked example: the binary symmetric channel}\label{sec:bsc-example}

Every object of this section takes an explicit closed form on the $n$-fold binary symmetric channel $\W=\mathrm{BSC}(p)^{\otimes n}$, $p<\tfrac12$, with $\cX=\cY=\{0,1\}^n$ and matched metric $m(x,y)=\log\W(y|x)$. The metric depends on $(x,y)$ only through the Hamming distance $d_H(x,y)$ and is strictly decreasing in it, so the decoder is nearest-neighbor and the analysis reduces to two binomial laws.

Fix the uniform input prior $Q_X=\Unif{\{0,1\}^n}$. A random competitor $\bar X\sim Q_X$ has $d_H(\bar X,y)\sim\mathrm{Bin}(n,\tfrac12)$ for every $y$; writing $b_k\triangleq\binom{n}{k}2^{-n}$ for its pmf and $B_k\triangleq\sum_{j<k}b_j$ for its strict CDF, the PEP of a candidate at distance $k=d_H(x,y)$ is, by Definition~\ref{def:pep}, the dithered binomial rank
\begin{equation}\label{eq:bsc-pep}
  \PEPU{x}{y}{u}=B_k+u\,b_k,\qquad k=d_H(x,y).
\end{equation}
Under the channel, $d_H(x,Y)$ equals the noise weight $K\sim\mathrm{Bin}(n,p)$, so the error spectrum is the explicit mixture
\begin{equation}\label{eq:bsc-spectrum}
  \Fspec{Q_X}{z}=\sum_{k=0}^{n}\binom{n}{k}p^{k}(1-p)^{n-k}\,
  \PR{B_k+U\,b_k\ge e^{-z}},\qquad U\sim\uU,
\end{equation}
continuous in $z$ by Lemma~\ref{lem:atomless}. The $\RCUp$ achievability~\eqref{eq:rcu-plus} is the binomial expectation $\tilde P_e(R;Q_X)=\Es{K,U}{\min\BRAs{1,e^{R}(B_K+U\,b_K)}}$, and the fixed-code converse~\eqref{eq:cc-converse} at $M=e^{R}$ codewords is the tail $P_e=\PR{B_K+U\,b_K\ge e^{-R}}=\Fspec{Q_X}{R}$: for the BSC the bracket~\eqref{eq:ach-conv-bracket} is the gap between integrating and thresholding one binomial spectrum. The uniform prior is exactly optimal in both coding directions at every blocklength (\cref{prop:symmetric}), so prior optimization is trivial for this channel; the non-symmetric case is what the programs of Sections~\ref{sec:prioropt} and~\ref{sec:ach-prioropt} address.

\subsection{The rate gap: quantifying the bracket}\label{sec:rategap}

We close the section by quantifying the bracket~\eqref{eq:ach-conv-bracket} at a \emph{fixed} prior. Fix $Q_X$, write $F(R)\triangleq\Fspec{Q_X}{R}$ for the converse side and $P(R)\triangleq\tilde P_e(R;Q_X)$ for the achievability side, and fix a target error level $\varepsilon\in(0,1)$: the converse permits rate $F^{-1}(\varepsilon)$, random coding achieves $P^{-1}(\varepsilon)$, and the question is how far apart the two rates are. The tools are the two identities already in hand --- the exponential-window form~\eqref{eq:exp-window} and the slope identity~\eqref{eq:slope-identity}. Write $\varsigma(R)\triangleq(\log F)'(R)=F'(R)/F(R)>0$ for the \emph{spectrum log-slope}: how fast the spectrum climbs, per nat, on a log scale.

\begin{definition}[Log-concave spectrum]\label{def:lc}
$F$ is \emph{log-concave at level $R$} if $\log F$ is concave on the whole upper tail $[R,\infty)$; equivalently, $\varsigma(\cdot)$ is non-increasing on $[R,\infty)$. (The window~\eqref{eq:exp-window} integrates the spectrum over all of $[R,\infty)$, so this is the domain on which the tangent bound below is used.)
\end{definition}

The assumption is the shape both asymptotic regimes suggest: in the large-deviations regime $-\log F\approx n\Lambda^\ast(R/n)$ with a convex Cram\'er rate function $\Lambda^\ast$, so $\log F$ is concave outright, and in the normal-approximation regime $F\approx\Phi\bigl((R-nI(Q_X,\W))/\sqrt{nV(Q_X,\W)}\bigr)$, log-concave because the normal CDF is, with $I(Q_X,\W)$ and $V(Q_X,\W)$ the per-letter mean and variance of the information density. \Cref{rem:onenat} quantifies both instances; the Gaussian shape is compared against the information-density spectrum in \cref{rem:gauss-regime}, and remains an assumption where it is used.

\begin{lemma}[Value sandwich]\label{lem:sandwich}
If $F$ is log-concave at level $R$ and $\varsigma(R)<1$, then $F(R)\le P(R)\le F(R)/(1-\varsigma(R))$.
\end{lemma}

\begin{IEEEproof}
The left inequality is the slope identity~\eqref{eq:slope-identity} ($\dot E\le0$). For the right, concavity of $u\mapsto\log F(R+u)$ on $[0,\infty)$ gives the tangent bound $F(R+u)\le F(R)\,e^{\varsigma(R)u}$ for every $u\ge0$; inserting it in~\eqref{eq:exp-window} and using $\varsigma(R)<1$, $P(R)\le F(R)\int_0^\infty e^{-(1-\varsigma(R))u}\,du=F(R)/(1-\varsigma(R))$.
\end{IEEEproof}

\begin{theorem}[Achievability--converse rate gap]\label{thm:rategap}
Fix $\varepsilon\in(0,1)$ and let $R_F=F^{-1}(\varepsilon)$, $R_P=P^{-1}(\varepsilon)$. Assume $F\in C^1$ is log-concave at level $R_P$ (\cref{def:lc}) with $\varsigma(R_P)<1$.
\begin{enumerate}
\item[(i)] \emph{(Non-increasing slope.)} If $\varsigma(R_F)>0$, then
\begin{equation}\label{eq:rategap-general}
0\;\le\;F^{-1}(\varepsilon)-P^{-1}(\varepsilon)\;\le\;\frac{-\log\bigl(1-\varsigma(R_P)\bigr)}{\varsigma(R_F)}.
\end{equation}
\item[(ii)] \emph{(Constant slope: the exponential boundary case.)} If moreover $\varsigma\equiv\varsigma_0\in(0,1)$ on $[R_P,R_F]$ (equivalently, $F$ is exactly exponential there), then
\begin{equation}\label{eq:rategap}
0\;\le\;F^{-1}(\varepsilon)-P^{-1}(\varepsilon)\;\le\;\frac{-\log(1-\varsigma_0)}{\varsigma_0}\;=\;1+\frac{\varsigma_0}{2}+\frac{\varsigma_0^2}{3}+\cdots
\end{equation}
\end{enumerate}
\end{theorem}

Of the $C^1$ hypothesis, continuity of $F$ is automatic (\cref{lem:atomless}); differentiability is assumed only so that the log-slope $\varsigma$ is defined.

\begin{IEEEproof}
The value sandwich (\cref{lem:sandwich}) at $R=R_P$ pins the spectrum's starting value, $F(R_P)\ge\varepsilon(1-\varsigma(R_P))$, from which it must climb to $F(R_F)=\varepsilon$. The change of variable $\rho=F(R)$ converts this climb into an integral of $1/(\varsigma\rho)$, which the non-increasing slope bounds by $-\log(1-\varsigma(R_P))/\varsigma(R_F)$; see Appendix~\ref{app:rategap} for the computation.
\end{IEEEproof}

\begin{remark}[The one-nat limit]\label{rem:onenat}
The genuine quantitative condition is $\varsigma<1$: at $\varsigma=1$ the window integral $\int_0^\infty e^{-(1-\varsigma)u}\,du$ diverges --- the straight-line regime $\dot E\to-1$ of~\eqref{eq:slope-identity}. As the spectrum flattens ($\varsigma_0\to0$) the behavior is two-faced: the \emph{value} gap vanishes, $1\le P/F\le1/(1-\varsigma_0)\to1$, yet the \emph{rate} gap need not --- the bound $-\log(1-\varsigma_0)/\varsigma_0$ extends continuously to the limit $1$. Heuristically, the spectrum must climb the vanishing log-height $-\log(1-\varsigma_0)\approx\varsigma_0$ at the vanishing log-slope $\varsigma_0$, and the ratio stays of order one: it is $\E{T}=1$, the mean of the window in~\eqref{eq:exp-window}, resurfacing.

In the two canonical regimes the corresponding approximations predict a limiting gap of one nat. For a locally Gaussian spectrum $F(R)\approx\Phi\bigl((R-nI)/\sqrt{nV}\bigr)$, with $I=I(Q_X,\W)$ and $V=V(Q_X,\W)$ as above --- log-concave since the normal CDF is, an assumed shape compared against the information-density spectrum in \cref{rem:gauss-regime} below --- the closed form gives a gap $1-\Phi^{-1}(\varepsilon)/(2\sqrt{nV})+O(1/n)\to1$. For a large-deviations spectrum $-\log F\approx n\Lambda^\ast(R/n)$, log-concavity is automatic by convexity of the Cram\'er rate function, $\varsigma=|{\Lambda^\ast}'|$ is the Chernoff tilt, and $\varsigma<1$ is exactly operating above the critical rate. Random coding and the converse can thus become indistinguishable in \emph{error probability} while remaining one nat apart in \emph{rate}: the unit forward shift of the kernel.
\end{remark}

\begin{remark}[Comparison with the information-density spectrum]\label{rem:gauss-regime}
Under the matched metric and any prior $Q_X$, the PEP is dominated pointwise by the information-density exponential:
\[
\PEPU{x}{y}{u}=G_{x,y}+u\,H_{x,y}\ \le\ G_{x,y}+H_{x,y}
=Q_X\BRA{\BRAs{\bar x:m(\bar x,y)\ge m(x,y)}}\ \le\ e^{-\idens{x}{y}},
\]
the last step by Markov's inequality: the event depends on the metric only through its ordering, so it may be evaluated with the likelihood-ratio metric $\W(y\mid x)/P_Y(y)$, whose expectation under $\bar X\sim Q_X$ is exactly one; here $\idens{x}{y}\triangleq\log\bigl(\W(y\mid x)/P_Y(y)\bigr)$ is the information density with respect to the induced output law $P_Y=Q_X\W$. This is the change-of-measure estimate already used after \cref{thm:rcu-plus}. Consequently
\[
\Fspec{Q_X}{R}\ \le\ \PR{\idens{X}{Y}\le R},\qquad(X,Y)\sim Q_X\cdot\W.
\]
For a product prior on a memoryless channel the information density is a sum of $n$ i.i.d.\ per-letter densities. Writing $I(Q_X,\W)$ for the mean of the per-letter density --- the mutual information at prior $Q_X$, taken with respect to the induced output law --- and $V(Q_X,\W)$ for its variance, abbreviated $I$ and $V$ here, the Berry--Esseen theorem gives the Gaussian form $\Phi\bigl((R-nI)/\sqrt{nV}\bigr)$ up to $O(1/\sqrt n)$. This comparison provides a Gaussian upper approximation to the PEP spectrum for product inputs. It does not by itself establish a Gaussian approximation or log-concavity for the PEP spectrum, which remain assumptions where they are used (\cref{def:lc}, \cref{thm:rategap}). Pushing the same domination through the $\RCUp$ kernel gives $\E{\min\BRAs{1,e^{R}e^{-\idens{X}{Y}}}}$, the dependence-testing-type bound noted after \cref{thm:rcu-plus}; at this level the resulting expansion of the achievable rate is second-order tight with an $O(1)$ third term. The standard $\tfrac12\log n$ third-order achievability term~\cite{polyanskiy2010channel} enters one level down, through the pairwise estimate itself: for non-singular memoryless channels a local central-limit sharpening of the Markov step above gives $\PEP{x}{y}\le B\,e^{-\idens{x}{y}}/\sqrt n$ on every output block satisfying a variance floor --- all but an $e^{-\Omega(n)}$-probability set (\cref{lem:pep-local}, stated and proved in Appendix~\ref{app:pep-local}; cf.~\cite[Lem.~47]{polyanskiy2010channel}). Substituting this estimate directly into the $\RCUp$ kernel, the exceptional set contributing at most its probability, gives
\[
\tilde P_e(R;Q_X)\ \le\ \E{\min\BRAs{1,e^{R-\frac12\log n+\log B}\,e^{-\idens{X}{Y}}}}+e^{-\Omega(n)}
\]
--- the same dependence-testing-type bound at the shifted rate $R-\tfrac12\log n+\log B$ --- whose Berry--Esseen analysis~\cite{polyanskiy2010channel} yields the usual $\tfrac12\log n$ third-order improvement under the stated non-singularity and moment conditions; the same substitution shifts the spectrum bound, $\Fspec{Q_X}{R}\le\PR{\idens{X}{Y}\le R-\tfrac12\log n+\log B}+e^{-\Omega(n)}$. Thus the local pairwise estimate exposes the $n^{-1/2}$ factor responsible for the third-order improvement.
\end{remark}

\section{Computing the Minimax Meta-Converse}\label{sec:prioropt}

The meta-converse identity of Section~\ref{sec:pep}, specialized to the matched metric~\eqref{eq:meta-channel} (restated as \cref{thm:cc-matched}), expresses the converse through the error spectrum $\Fspec{Q_X}{R}$, whose value depends on the input prior $Q_X$. The \emph{tightest} converse over priors is the minimax meta-converse
\begin{equation}\label{eq:metaconverse-intro}
	\inf_{Q_X\in\cP(\cX)}\Fspec{Q_X}{R}
	\;=\;\inf_{Q_X}\max_{Q_Y}\;\BETA{1-e^{-R}}{Q_X\times Q_Y}{Q_X\W}.
\end{equation}
Under matched maximum-likelihood decoding $Q_X\mapsto\Fspec{Q_X}{R}$ is convex (Theorem~\ref{thm:joint-convex}), so~\eqref{eq:metaconverse-intro} is a convex program; but as written it is a genuine minimax, the inner $\sup_{Q_Y}$ opposing the outer $\inf_{Q_X}$.

The contribution of this section is that the \emph{reverse-channel identity} of Section~\ref{sec:pep} removes this opposition and collapses~\eqref{eq:metaconverse-intro} to a single finite-dimensional linear program, in the input prior and a reverse channel jointly; for memoryless channels with fixed alphabets, a type reduction (\cref{prop:lp-type}) makes the program polynomial in the blocklength. The program is the finite-blocklength instance of Matthews' nonsignaling converse~\cite{matthews2012linear}, reached here directly from the PEP spectrum rather than through information-theoretic relaxations. Throughout, alphabets are finite.

\subsection{Aligned directions via the reverse channel}\label{sec:lp}

Recall (Theorem~\ref{thm:reverse-channel}) that for the likelihood-ratio metric $m(x,y)=\W(y|x)/Q_Y(y)$ with $Q_Y$ of full support --- order-equivalent in $x$ to the matched metric, hence inducing the same spectrum, and finite with finite mean on the finite alphabets, so the hypotheses of Theorem~\ref{thm:reverse-channel} hold ---
\begin{equation}\label{eq:rc-sup-channel}
	\Es{Q_X\times Q_Y}{m(X,Y)\,\Ind{\PEP{X}{Y}\le e^{-R}}}
	=\!\!\sup_{\substack{W_{X|Y}^*:\;W_{X|Y}^*(x|y)\le e^R Q_X(x)}}\!\! e^{-R}\,\Es{W_{X|Y}^*Q_Y}{m(X,Y)},
\end{equation}
where $W_{X|Y}^*$ ranges over reverse channels (kernels from $\cY$ to $\cX$) and the conditional $D_\infty$ cap $D_\infty(W_{X|Y}^*\|Q_X\mid Q_Y)\le R$ reads $W_{X|Y}^*(x|y)\le e^R Q_X(x)$.

By atomlessness of the dithered PEP under absolute continuity (Lemma~\ref{lem:atomless}), the minimax converse is a complementary supremum,
\begin{equation*}
	1-\inf_{Q_X}\Fspec{Q_X}{R}=\sup_{Q_X}\PR{\PEP{X}{Y}\le e^{-R}}.
\end{equation*}
Substituting $m(x,y)=\W(y|x)/Q_Y(y)$ into~\eqref{eq:rc-sup-channel} and taking the supremum over $Q_X$,
\begin{align*}
	\sup_{Q_X}\PR{\PEP{X}{Y}\le e^{-R}}
	&=\sup_{Q_X}\ \sup_{W_{X|Y}^*(x|y)\le e^R Q_X(x)}\!e^{-R}\sum_{x,y}W_{X|Y}^*(x|y)\,Q_Y(y)\,\frac{\W(y|x)}{Q_Y(y)}\\
	&=\sup_{\substack{Q_X,\,W_{X|Y}^*:\;W_{X|Y}^*(x|y)\le e^R Q_X(x)}}e^{-R}\sum_{x,y}W_{X|Y}^*(x|y)\,\W(y|x).
\end{align*}
The auxiliary law $Q_Y$ cancels in the objective and is absent from the constraints, so the objective is \emph{linear} in the joint decision variables $(Q_X,W_{X|Y}^*)$, and the two optimizations --- now both suprema, no longer opposed --- merge into a single program.

\begin{theorem}[Minimax meta-converse as a linear program]\label{thm:lp}
	The prior-optimized minimax meta-converse satisfies
	\begin{equation*}
		1-\inf_{Q_X}\Fspec{Q_X}{R}=\text{\emph{(optimal value of the LP)}}
	\end{equation*}
	\begin{align}
		\text{\emph{maximize}}\quad & e^{-R}\sum_{x,y}W_{X|Y}^*(x|y)\,\W(y|x)\label{eq:lp-obj}\\
		\text{\emph{s.t.}}\quad
		& \textstyle\sum_x Q_X(x)=1,\quad Q_X(x)\ge0,\quad\forall x,\label{eq:lp-c1}\\
		& \textstyle\sum_x W_{X|Y}^*(x|y)=1,\ \ W_{X|Y}^*(x|y)\ge0,\quad\forall x,y,\label{eq:lp-c2}\\
		& W_{X|Y}^*(x|y)\le e^R\,Q_X(x),\quad\forall x,y,\label{eq:lp-c3}
	\end{align}
	a linear program in $(Q_X,W_{X|Y}^*)$ with $O(|\cX|\,|\cY|)$ variables and constraints.
\end{theorem}

\begin{IEEEproof}
	The objective and constraints are linear by the cancellation of $Q_Y$ above: \eqref{eq:lp-c2} requires $W_{X|Y}^*$ to be a stochastic kernel and \eqref{eq:lp-c3} is the conditional $D_\infty$ cap of~\eqref{eq:rc-sup-channel}. The optimal value therefore equals the joint supremum displayed above, which is $1-\inf_{Q_X}\Fspec{Q_X}{R}$.
\end{IEEEproof}

The variables $W_{X|Y}^*(x|y)$ are acceptance probabilities: the chance that input $x$ is declared correct given output $y$. Constraint~\eqref{eq:lp-c2} accepts exactly one input per output in expectation, and~\eqref{eq:lp-c3} couples the inputs competing at a common output through the shared prior $Q_X$. The program has $O(|\cX|\,|\cY|)$ variables; standard interior-point solvers apply. In the memoryless setting the program inherits the coordinate-permutation symmetry of the channel and collapses onto types (\cref{prop:lp-type} below), which is what makes finite-blocklength evaluation tractable (Section~\ref{sec:num}).

\paragraph{Connection to Matthews' converse} The LP~\eqref{eq:lp-obj}--\eqref{eq:lp-c3} is precisely the finite-blocklength instance of the nonsignaling converse of Matthews~\cite{matthews2012linear}: the reverse-channel variables $W_{X|Y}^*$ play the role of nonsignaling ``response functions'', and the conditional $D_\infty$ cap~\eqref{eq:lp-c3} is exactly the nonsignaling constraint; the LP-converse line that followed Matthews was developed further by Jose and Kulkarni~\cite{jose2017linear}, who extended linear-programming converses to finite-blocklength lossy joint source--channel coding.

The two routes are complementary. Matthews \emph{starts} from a relaxation of the code structure (nonsignaling codes) and \emph{arrives} at an LP converse; we start from the PEP error spectrum and the reverse-channel variational identity (Theorem~\ref{thm:reverse-channel}) and arrive at the \emph{same} LP, exactly equivalent to the minimax meta-converse. The reverse channel $W_{X|Y}^*(x|y)$ is the dithered posterior ``acceptance'' probability that input $x$ is the transmitted one given output $y$, capped at $e^{R}Q_X(x)$ by the spectrum's threshold structure. The PEP derivation thus gives the nonsignaling LP an operational reading and places it on the same footing as the achievability program of Section~\ref{sec:ach-prioropt}. The identification is exact, not merely structural, as the LP dual shows.

\begin{remark}[The LP dual: per-output thresholds]\label{rem:lp-dual}
Attach multipliers to the constraints of~\eqref{eq:lp-obj}--\eqref{eq:lp-c3} and eliminate those of the box~\eqref{eq:lp-c3} at their smallest feasible values: the dual collapses to a maximin over per-output thresholds $t\in[0,1]^{|\cY|}$, and strong LP duality reads
\begin{equation}\label{eq:lp-dual}
\inf_{Q_X}\Fspec{Q_X}{R}
=\max_{t\in[0,1]^{|\cY|}}\ \min_{x\in\cX}\ \Bigl[\sum_y\min\bigl\{\W(y\mid x),\,t_y\bigr\}-e^{-R}\sum_y t_y\Bigr],
\end{equation}
the optimal multipliers of the per-output normalization rows~\eqref{eq:lp-c2} being $\nu_y^\star=e^{-R}t_y^\star$. The thresholds are water levels: each output's column of the channel is clipped at $t_y$, so the inner minimand is a per-output water-filling account, and complementary slackness on the prior rows reproduces the equal-marginal-value optimality conditions of the achievability side (\cref{prop:ach-kkt}). Since the primal value is $1-\inf_{Q_X}\Fspec{Q_X}{R}$ (\cref{thm:lp}), \eqref{eq:lp-dual} identifies Matthews' LP value with the minimax meta-converse exactly. The eliminate-and-substitute algebra is carried out in Appendix~\ref{app:lp-dual}.
\end{remark}

\subsection{Memoryless channels: reduction to types}\label{sec:lp-type}

As written, for a memoryless channel $\W=W^{\otimes n}$ the LP lives on the product alphabets, with $|\cX|^n$ prior variables and $|\cX|^n|\cY|^n$ reverse-channel variables. The coordinate-permutation symmetry of the product channel collapses it to polynomial size.

\begin{proposition}[Type reduction of the converse LP]\label{prop:lp-type}
Let the channel be memoryless, $\W=W^{\otimes n}:\cX^n\to\cY^n$ with $\cX,\cY$ finite. Then the LP of \cref{thm:lp} admits an optimal solution $(Q_X,W_{X|Y}^*)$ invariant under simultaneous coordinate permutations: $Q_X(x)$ depends on $x\in\cX^n$ only through its type, and $W_{X|Y}^*(x\mid y)$ depends on $(x,y)$ only through the joint type of the pair. The program therefore reduces to $\binom{n+|\cX|-1}{|\cX|-1}$ prior variables and $\binom{n+|\cX||\cY|-1}{|\cX||\cY|-1}$ reverse-channel variables --- polynomial in $n$ at fixed alphabets, of degrees $|\cX|-1$ and $|\cX||\cY|-1$ respectively. The per-output-threshold dual~\eqref{eq:lp-dual} symmetrizes likewise: an optimal threshold vector exists whose entries $t_y$ depend only on the type of $y$.
\end{proposition}

\begin{IEEEproof}
For a permutation $\pi$ of the $n$ coordinates write $(\pi x)_i\triangleq x_{\pi^{-1}(i)}$; memorylessness gives $\W(\pi y\mid\pi x)=\W(y\mid x)$. Let $\pi$ act on solutions by
\[
Q_X^\pi\triangleq Q_X\circ\pi^{-1},\qquad
W^{*,\pi}(x\mid y)\triangleq W_{X|Y}^*\bigl(\pi^{-1}x\mid\pi^{-1}y\bigr),
\]
the substitution used (for a general symmetry group) in \cref{rem:group-symmetric} below. This action preserves the normalizations~\eqref{eq:lp-c1}--\eqref{eq:lp-c2}, maps the box~\eqref{eq:lp-c3} to $W^{*,\pi}(x\mid y)=W_{X|Y}^*(\pi^{-1}x\mid\pi^{-1}y)\le e^R Q_X(\pi^{-1}x)=e^R Q_X^\pi(x)$, and preserves the objective~\eqref{eq:lp-obj} by the substitution $(x,y)\mapsto(\pi x,\pi y)$ together with $\W(\pi y\mid\pi x)=\W(y\mid x)$. The feasible set is convex (all constraints are linear in the pair $(Q_X,W_{X|Y}^*)$) and the objective is linear, so the average of an optimal solution over all $n!$ permutations is feasible, attains the same optimal value, and is invariant under every $\pi$. The orbits of the action are the type classes of $x$ on $\cX^n$ and, under the simultaneous action on $\cX^n\times\cY^n$, the joint type classes of $(x,y)$; an invariant solution is constant on orbits, and the orbit counts are the stated binomial coefficients. For the dual, the objective of~\eqref{eq:lp-dual} is concave in $t\in[0,1]^{|\cY|^n}$ and invariant under $t\mapsto t\circ\pi^{-1}$ (reindex the inner minimization by $x\mapsto\pi x$), so the same orbit average of an optimal $t$ is optimal and constant on output types.
\end{IEEEproof}

Channel structure (zero entries, additional symmetries) can shrink the reduced program further; \cref{sec:num} reports the sizes used in practice. The achievability-side counterpart of this reduction --- the march run on the type simplex --- is \cref{rem:ach-type}.

The same reverse-channel construction applies, over an enlarged candidate space, to lossy source coding and joint source--channel coding; those directions are not pursued here. Channel symmetry, by contrast, collapses the program outright, as we record next.

\subsection{Symmetric channels: a closed-form optimum}\label{sec:symmetric}

When the channel is symmetric enough the program need not be solved at all: the uniform prior is optimal, in both coding directions and at every blocklength. Call $\W$ \emph{Gallager-symmetric}~\cite[p.~94]{gallager1968information} if the output alphabet partitions into blocks $\cY=\cY_1\cup\cdots\cup\cY_K$ such that within each block $\cY_i$ the transition submatrix $[\W(y\mid x)]_{x\in\cX,\,y\in\cY_i}$ has rows that are permutations of one another and likewise columns: the multiset $\{\W(y\mid x):y\in\cY_i\}$ is the same for every $x\in\cX$, and the multiset $\{\W(y\mid x):x\in\cX\}$ is the same for every $y\in\cY_i$.

\begin{proposition}[Uniform prior for Gallager-symmetric channels]\label{prop:symmetric}
Let $\W$ be Gallager-symmetric and write $U\triangleq\Unif{\cX}$. Then the uniform prior attains the minimax meta-converse, $\Fspec{U}{R}=\inf_{Q_X}\Fspec{Q_X}{R}$, for every $R\ge0$, and likewise maximizes the achievability objective $J$ of \cref{prop:ach-concave} (\cref{sec:ach-prioropt} below) for every nonnegative kernel.
\end{proposition}

\begin{IEEEproof}
It suffices that $U$ maximize the PEP-scale spectrum $G(\cdot;w)$ of \cref{lem:waterfill} at every $w\in(0,1]$: then $J(Q_X)=\int_0^1 G(Q_X;w)\,\kappa(w)\,dw\le J(U)$ for every kernel $\kappa\ge0$, and since the two spectra are complements under $w=e^{-R}$ (\cref{lem:atomless}), $\Fspec{U}{R}=1-G(U;e^{-R})\le 1-G(Q_X;e^{-R})=\Fspec{Q_X}{R}$ for every $Q_X$.

Fix $w\in(0,1]$, let $n\triangleq\abs{\cX}$ and $m\triangleq\#\{j:j/n<w\}\le n-1$, and for each $y$ let $t_y$ be the $(m{+}1)$-th largest of the column entries $\{\W(y\mid x):x\in\cX\}$; by the column condition, $t_y=t_i$ is constant on each block $\cY_i$. In the notation of \cref{lem:waterfill}, the fill orders are prior-independent, $t_y=\nu^y_{m+1}$, and the uniform cumulative masses are $\sigma^y_j(U)=j/n$ for every $y$, so $\sigma^y_j(U)<w$ exactly when $j\le m$. The elementary ramp inequality $\min\{w,\sigma\}\le\min\{w,\sigma_0\}+\Ind{\sigma_0<w}\,(\sigma-\sigma_0)$, applied at $\sigma_0=\sigma^y_j(U)$ term by term in the ramp form $G(Q_X;w)=\sum_y\sum_j c^y_j\min\{w,\sigma^y_j(Q_X)\}$ of~\eqref{eq:wf-ramp}, gives
\begin{equation*}
G(Q_X;w)\;\le\;G(U;w)+\sum_{y}\sum_{j\le m}c^y_j\bigl(\sigma^y_j(Q_X)-\sigma^y_j(U)\bigr)
\;=\;G(U;w)+\sum_{x}\bigl(Q_X(x)-\tfrac1n\bigr)\,g(x),
\end{equation*}
where the second equality exchanges sums via $\sigma^y_j(Q_X)=\sum_{i\le j}Q_X(x^y_{(i)})$ and telescopes the increments: writing $j(x,y)$ for the rank of $x$ in output $y$'s fill order, $\sum_{j(x,y)\le j\le m}c^y_j=\nu^y_{j(x,y)}-\nu^y_{m+1}=\W(y\mid x)-t_y$ when $j(x,y)\le m$ and the sum is empty otherwise, so
\begin{equation*}
g(x)=\sum_{y}\bigl(\W(y\mid x)-t_y\bigr)^{+}.
\end{equation*}
(Entries above $t_y$ have rank at most $m$, entries below it rank above $m$, and ties are harmless: consecutive tied levels have $c^y_j=0$, so an entry tied with $t_y$ contributes $t_y-t_y=0$ whichever side of $m$ its rank falls.) By the row condition, $\sum_{y\in\cY_i}(\W(y\mid x)-t_i)^{+}$ depends only on the multiset of row entries within the block, hence not on $x$; thus $g\equiv\lambda$ is constant, the correction term is $\lambda\sum_x(Q_X(x)-\tfrac1n)=0$, and $G(Q_X;w)\le G(U;w)$ for every $Q_X$.
\end{IEEEproof}

In both directions prior optimization is therefore unnecessary: the bound is the spectrum read --- or integrated against the kernel --- at the uniform prior.

\begin{remark}[Transitive group symmetry]\label{rem:group-symmetric}
A transitive symmetry group gives a shorter route to the same conclusion. Suppose a group $\Gamma$ acts transitively on $\cX$ and, for each $g\in\Gamma$, there is a permutation $\sigma_g$ of $\cY$ with $\W(\sigma_g(y)\mid g(x))=\W(y\mid x)$ for all $x,y$. Such a channel is Gallager-symmetric: on each orbit of the output-permutation group generated by $\{\sigma_g\}$, the relabelings $(x,y)\mapsto(g(x),\sigma_g(y))$ carry any row multiset to any other (transitivity) and preserve every column multiset, so the orbits serve as blocks. For this subclass an orbit-averaging argument replaces the certificate: each objective is concave in $Q_X$ (linear, for the converse LP) and invariant under $Q_X\mapsto Q_X\circ g^{-1}$ --- for the LP via the substitution $(Q_X,W_{X|Y}^*)\mapsto(Q_X\circ g^{-1},\,W_g^*)$ with $W_g^*(x\mid y)=W_{X|Y}^*\bigl(g^{-1}(x)\mid\sigma_g^{-1}(y)\bigr)$, which preserves the constraints~\eqref{eq:lp-c1}--\eqref{eq:lp-c3} and the objective~\eqref{eq:lp-obj}; for $J$ because the water-filling value at $\sigma_g(y)$ under $Q_X\circ g^{-1}$ equals that at $y$ under $Q_X$, so $G(Q_X\circ g^{-1};w)=G(Q_X;w)$ at every $w$ --- and averaging an optimum over its orbit $\{Q_X^\star\circ g^{-1}:g\in\Gamma\}$ yields, by concavity, a $\Gamma$-invariant optimum, of which transitivity leaves only the uniform prior. The inclusion is strict: there are Gallager-symmetric channels with trivial automorphism group --- blocks built from Latin squares, say --- for which no averaging argument is available and \cref{prop:symmetric} still applies.
\end{remark}

For the binary-symmetric channel both routes apply: the matrix is Gallager-symmetric with the single block $\cY=\{0,1\}$ (every row and column a permutation of $(1-p,\,p)$), and $\Gamma=\{\mathrm{id},\,\text{flip}\}$ acts transitively on $\{0,1\}$ with $\sigma$ the output bit-flip. Either way the uniform input is optimal in both directions at every blocklength --- the closed-form fact behind the analytic illustration of \cref{sec:bsc-example}.

\section{Prior Optimization of the Achievability Bound}\label{sec:ach-prioropt}

Section~\ref{sec:prioropt} optimized the prior on the \emph{converse} side, where
the spectrum is read at a \emph{single} threshold and the minimax collapses to one
linear program. The \emph{achievability} bound is, by the integral
representation~\eqref{eq:integral}, a positive-kernel integral of the \emph{whole}
spectrum, so its prior optimization couples all thresholds at once. We show that
the two ingredients already in hand --- the integral
formula~\eqref{eq:integral} and the reverse-channel
identity~\eqref{eq:rc-sup} --- make this a \emph{concave} program over the
input simplex.  Its
gradient is read off the spectrum in one sweep, so the program is solved by a
direct first-order march.

The converse
reads the spectrum at one threshold while achievability integrates it against a
kernel: the two share a common prior-optimization structure --- the input
simplex, ordered by the channel's water-filling --- seen through two kernels,
the converse's indicator and the random-coding kernel.

\subsection{The bound as a kernel integral}

For an input prior $Q_X$ write the non-decreasing spectrum in the PEP variable
$w\in[0,1]$,
\begin{equation}\label{eq:ach-spectrum}
G(Q_X;w)\triangleq\PR{\PEP{X}{Y}\le w},\quad G(Q_X;0)=0,\quad(X,Y)\sim Q_X\cdot\W;
\end{equation}
we write $G$ for this PEP-scale CDF, reserving $F$ for the rate-scale spectrum
$\Fspec{Q_X}{z}$ of \cref{def:spectrum} --- the two are complements under
$w=e^{-z}$. Both random-coding bounds of Section~\ref{sec:channel} then have the form
\begin{equation}\label{eq:ach-J}
J(Q_X)\triangleq\int_0^1 G(Q_X;w)\,\kappa(w)\,dw,\qquad P_e=1-J(Q_X),
\end{equation}
for a nonnegative kernel $\kappa$: the exact bound
(\cref{thm:exact-rc}) uses $\kappa(w)=(M-1)(1-w)^{M-2}$ and the
$\RCUp$ bound (\cref{thm:rcu-plus}) $\kappa(w)=e^{R}\Ind{w\le e^{-R}}$, both
obtained by Stieltjes integration by parts (with $G(Q_X;0)=0$, $G(Q_X;1)=1$).
Minimizing $P_e$ is maximizing $J$. Throughout this section we use the
achievability-natural convention $M-1=e^{R}$ for the random-coding ensemble; the
bound at the integer codebook size $M=\lceil e^{R}\rceil$ follows by monotonicity
in the codebook size. The dithered PEP is atomless (\cref{lem:atomless}), so
$G(Q_X;\cdot)$ is continuous.

\subsection{Water-filling the spectrum}

At the matched (likelihood-ratio) metric the reverse-channel identity gives the
spectrum, at each threshold, as a fractional-knapsack value whose fill order is
fixed by the channel.

\begin{lemma}[Water-filling form]\label{lem:waterfill}
For the likelihood-ratio metric $m(x,y)=\W(y\mid x)/Q_Y(y)$ ($Q_Y$ of full
support, as in \cref{sec:lp}) and every $w\in[0,1]$,
\begin{equation}\label{eq:waterfill}
G(Q_X;w)=\sum_{y\in\cY}s_y(Q_X;w),\qquad
s_y(Q_X;w)=\max_{\substack{0\le V\le Q_X\\ \sum_x V(x)=w}}\ \sum_{x}V(x)\,\W(y\mid x).
\end{equation}
Fix $y$ and order $\cX$ as $x^{y}_{(1)},\dots,x^{y}_{(n)}$ ($n=\abs{\cX}$) so that
$\W(y\mid x^{y}_{(1)})\ge\cdots\ge\W(y\mid x^{y}_{(n)})$ (ties broken arbitrarily).
With $\nu^{y}_{j}\triangleq\W(y\mid x^{y}_{(j)})$, $\nu^{y}_{n+1}\triangleq0$, the
metric increments $c^{y}_{j}\triangleq\nu^{y}_{j}-\nu^{y}_{j+1}\ge0$, and the
cumulative masses $\sigma^{y}_{j}(Q_X)\triangleq\sum_{i\le j}Q_X(x^{y}_{(i)})$,
$\sigma^{y}_{0}\triangleq0$,
\begin{equation}\label{eq:wf-ramp}
s_y(Q_X;w)=\sum_{j}c^{y}_{j}\,\min\bigl\{w,\sigma^{y}_{j}(Q_X)\bigr\},
\end{equation}
equivalently the staircase integral
\begin{equation}\label{eq:wf-staircase}
s_y(Q_X;w)=\int_{0}^{w}\rho_y(Q_X;t)\,dt,\qquad
\rho_y(Q_X;t)=\nu^{y}_{j}\ \text{ for } t\in\bigl(\sigma^{y}_{j-1},\sigma^{y}_{j}\bigr].
\end{equation}
The fill order $x^{y}_{(1)},\dots,x^{y}_{(n)}$ depends only on the channel, not on
$Q_X$; the breakpoints $\sigma^{y}_{j}$ are linear in $Q_X$; and each
$s_y(Q_X;\cdot)$ is non-decreasing, concave and piecewise-linear in $w$.
\end{lemma}

\begin{IEEEproof}
The change of variable $V=w\,W^{*}$ turns the reverse-channel cap into the box
$V\le Q_X$ with budget $w$; the constraints decouple across outputs, and each
per-output program is a fractional knapsack whose optimum fills inputs in
decreasing channel-value order --- the ramp form~\eqref{eq:wf-ramp} follows by
summing level increments (Appendix~\ref{app:ach-prioropt}).
\end{IEEEproof}

\begin{remark}[One program, two kernels]\label{rem:one-program}
Writing $A_y$ for the linear map taking the prior to its cumulative masses
$\bigl(\sigma^{y}_{j}(Q_X)\bigr)_j$ and $c_y\triangleq(c^{y}_{j})_j\ge0$ for the
metric increments, \eqref{eq:wf-ramp} is the composition
$s_y(Q_X;w)=c_y^{\top}\Phi_w(A_yQ_X)$ with the ramp potential
$\Phi_w(\sigma)\triangleq\min\{w,\sigma\}$ applied componentwise; hence
\[
G(Q_X;w)=\sum_{y}c_y^{\top}\Phi_w\bigl(A_y Q_X\bigr),\qquad
J(Q_X)=\sum_{y}c_y^{\top}\Phi_\kappa\bigl(A_y Q_X\bigr),
\]
with the kernel-smoothed average
$\Phi_\kappa(\sigma)\triangleq\int_0^1\min\{w,\sigma\}\,\kappa(w)\,dw$. The
converse and achievability prior optimizations are thus \emph{one} program
$c^\top\Phi(AQ_X)$ distinguished only by the potential $\Phi$: the converse LP
of \cref{thm:lp} maximizes the single ramp $\Phi_{e^{-R}}$ (its value being
$1-\inf_{Q_X}\Fspec{Q_X}{R}$), while achievability maximizes the smoothed
$\Phi_\kappa$.
\end{remark}

\subsection{Direct optimization on the simplex}

The water-filling form makes the prior program concave and, at the same time,
hands over its gradient: nothing more is needed to march to the optimum.

\begin{proposition}[Concave prior program and its gradient]\label{prop:ach-concave}
For every nonnegative kernel $\kappa$ the achievability objective
$Q_X\mapsto J(Q_X)=\int_0^1 G(Q_X;w)\,\kappa(w)\,dw$ of~\eqref{eq:ach-J} is
concave on the simplex $\cP(\cX)$, so prior optimization of the random-coding
bound is the concave maximization
\begin{equation}\label{eq:ach-concave}
\inf_{Q_X\in\cP(\cX)}P_e(R;Q_X)=1-\max_{Q_X\in\cP(\cX)}J(Q_X).
\end{equation}
Writing $\bar\kappa(u)\triangleq\int_u^1\kappa(w)\,dw$ and using the fill order,
unit values $\nu^{y}_{j}$ (with $\nu^{y}_{n+1}\triangleq0$), and cumulative masses
$\sigma^{y}_{j}(Q_X)$ of \cref{lem:waterfill}, the gradient is
\begin{equation}\label{eq:ach-grad}
\frac{\partial J}{\partial Q_X(x)}
=\sum_{y\in\cY}\ \sum_{j\ge j(x,y)}\bigl(\nu^{y}_{j}-\nu^{y}_{j+1}\bigr)\,
\bar\kappa\bigl(\sigma^{y}_{j}(Q_X)\bigr),
\end{equation}
where $j(x,y)$ is the rank of $x$ in output $y$'s fill order; a single sweep over
the $\abs{\cX}\abs{\cY}$ knots evaluates $J$ and $\nabla J$ together.
\end{proposition}

\begin{IEEEproof}
By \cref{lem:waterfill} each per-output value $s_y(Q_X;w)$ is the optimum of a
linear program whose box constraint $V\le Q_X$ carries the prior on its
right-hand side; the optimal value of such a program is concave in the
right-hand side, so $s_y(Q_X;\cdot)$ is concave in $Q_X$, and hence so are
$G=\sum_y s_y$ and, for $\kappa\ge0$, the integral $J$. For the gradient,
substituting the staircase form~\eqref{eq:wf-staircase} and exchanging the order
of integration turns $J$ into knot sums with breakpoints linear in the prior;
differentiating breakpoint by breakpoint gives~\eqref{eq:ach-grad}
(Appendix~\ref{app:ach-prioropt}).
\end{IEEEproof}

\begin{proposition}[Water-filling optimality]\label{prop:ach-kkt}
A prior $Q_X^{\star}$ maximizes $J$ over $\cP(\cX)$ if and only if there is a level
$\lambda$ such that the marginal value
$g^{\star}(x)\triangleq\partial J/\partial Q_X(x)\big|_{Q_X^{\star}}$ satisfies
$g^{\star}(x)=\lambda$ for every $x$ in the support of $Q_X^{\star}$ and
$g^{\star}(x)\le\lambda$ otherwise.
\end{proposition}

\begin{IEEEproof}
This is the standard KKT condition for maximizing a concave function over the
simplex (\cref{prop:ach-concave}); concavity makes it sufficient.
\end{IEEEproof}

\Cref{prop:ach-kkt} is the achievability analogue of the converse's complementary
slackness: the optimal prior equalizes the marginal value across its support, the
channel's water-filling order fixing how that value accrues at each output.
Because $J$ is concave and \cref{prop:ach-concave} returns $\nabla J$ in one
sweep, the optimum is reached by a direct first-order march on the simplex ---
projected gradient, mirror descent, or Frank--Wolfe. The Frank--Wolfe step is the
most transparent: its linear oracle $\max_{Q\in\cP(\cX)}\sum_x g(x)\,Q(x)$ places
all mass on the input of largest marginal value, so every step is itself a
water-filling move.

The kernel enters the gradient~\eqref{eq:ach-grad} only
through the scalar weights $\bar\kappa\bigl(\sigma^{y}_{j}\bigr)$ --- evaluations
of the tail integral $\bar\kappa(\cdot)$ at the knots --- so the same solver
applies to every nonnegative kernel of \cref{prop:ach-concave}: the exact random-coding kernel $\kappa(w)=(M-1)(1-w)^{M-2}$ and the
$\RCUp$ kernel $\kappa(w)=e^{R}\Ind{w\le e^{-R}}$ are handled with no
kernel-specific reformulation.

\begin{proposition}[Convergence of the march]\label{prop:ach-march}
Assume the matched decoding metric (the likelihood-ratio metric of
\cref{lem:waterfill}) and either random-coding kernel --- both are
nonincreasing with maximum $\kappa(0^{+})=e^{R}$ under the convention
$M-1=e^{R}$. Then the gradient~\eqref{eq:ach-grad} is Lipschitz on the simplex:
\begin{equation}\label{eq:ach-smooth}
\norm{\nabla J(Q_X)-\nabla J(Q_X')}_{\infty}\le L\,\norm{Q_X-Q_X'}_{1},
\qquad L=\kappa(0^{+})=e^{R}.
\end{equation}
Consequently the Frank--Wolfe march
$Q_X^{(k+1)}=(1-\gamma_k)\,Q_X^{(k)}+\gamma_k\,\delta_{x_k}$, with
$x_k\in\argmax_x g(x;Q_X^{(k)})$ and $\gamma_k=2/(k+2)$, started anywhere on
$\cP(\cX)$, satisfies
\begin{equation}\label{eq:ach-rate}
\max_{Q_X\in\cP(\cX)}J(Q_X)-J\bigl(Q_X^{(k)}\bigr)\le\frac{8L}{k+2},
\end{equation}
each iteration costing one $O(\abs{\cX}\abs{\cY})$ gradient sweep.
\end{proposition}

\begin{IEEEproof}
The cumulative masses $\sigma^{y}_{j}(Q_X)$ are $1$-Lipschitz in the prior and
$\bar\kappa$ is $\kappa(0^{+})$-Lipschitz; for the matched metric the telescoped
weights sum over outputs to one, giving $L=e^{R}$, and the standard
Frank--Wolfe guarantee on the diameter-$2$ simplex yields the $O(L/k)$ rate
(Appendix~\ref{app:ach-prioropt}).
\end{IEEEproof}

The smoothness constant $L=e^{R}$ is a worst-case bound, exponential in the
rate, so the guarantee~\eqref{eq:ach-rate} is loose at practical rates; it
certifies convergence, not its observed speed. In the experiments of
\cref{sec:num} the march, with exact line search and warm starts, reaches
certificate accuracy within a few hundred iterations per rate point --- far
fewer than \eqref{eq:ach-rate} predicts.

\begin{remark}[Memoryless channels: a polynomial march]\label{rem:ach-type}
For a memoryless channel $\W^{\otimes n}$ the objective $J$ is invariant under a
common permutation of the $n$ coordinates, so by concavity it is maximized by a
permutation-invariant prior --- a mixture of uniform laws over the type classes.
The march may therefore run on the \emph{type simplex}, one variable per
composition, of dimension $\binom{n+\abs{\cX}-1}{\abs{\cX}-1}$ --- polynomial in
$n$ at fixed alphabet, the prior-side instance of the type reduction that makes
the converse LP tractable (\cref{prop:lp-type}). The fill order and the gradient~\eqref{eq:ach-grad} are
evaluated on the induced type-level channel, so each step stays polynomial
in~$n$.
\end{remark}

\section{Numerical Illustration}
\label{sec:num}

Two examples illustrate the two operational consequences of the framework: the converse--achievability \emph{bracket} produced by evaluating one and the same pairwise object in both directions (Section~\ref{sec:channel}), and the \emph{prior-family} improvement obtained by replacing a memoryless input law with the achievability-optimal type-based prior, computed by the simplex march of Section~\ref{sec:ach-prioropt}.

They complement the fully analytic binary-symmetric-channel illustration of Section~\ref{sec:bsc-example}, where both directions reduce in closed form to one binomial spectrum and prior optimization is rendered trivial by symmetry. The two numerical channels here exercise the cases the BSC does not --- a continuous-alphabet channel evaluated by exact non-central-$t$ computation, and an asymmetric channel where the prior genuinely matters. We focus on these two representative channels; additional channels and rate--SNR sweeps follow by the same exact evaluation. Code reproducing both experiments --- including the simplex march with per-point optimality certificates and its verification against an independent exact convex-programming solve --- is available in the companion libraries~\cite{elkayam2026oneshotIT,elkayam2026awgnfbl}.

\subsection{AWGN: an exact converse--achievability bracket}

For the additive white Gaussian noise channel under maximum-likelihood decoding, with codewords on the power shell $\BRAs{\bx\colon\norm{\bx}^2=nP}$, the pairwise error probability admits a closed form: conditioned on the angular correlation $T=\BRAi{\bx,\by}/(\norm{\bx}\norm{\by})$ between the transmitted codeword and the received word, the per-codeword pairwise error is the deterministic function $G(t)=\PR{\hat\rho\ge t}$ of an isotropic competitor against a fixed direction. This single function plays both roles of the framework.

As the $\beta$-functional it gives the optimal meta-converse for the uniform-on-shell input: among all auxiliary output measures $Q_{Y^n}$ the rotationally symmetric choice is $\beta$-minimizing~\cite{polyanskiy2010thesis,polyanskiy2013saddle}, the likelihood-ratio test reduces to a spherical cap, and $\beta_{1-\varepsilon}$ equals $G(t)$ with the threshold $t$ fixed by the type-I constraint through the non-central-$t$ CDF of $T$. This recovers Shannon's 1959 sphere-packing bound~\cite{shannon1959probability} as an instance of the meta-converse (\cref{thm:meta-converse}), here evaluated in exact closed form.

As the achievability kernel the same $G$ enters the $\RCUp$ integral (\cref{thm:rcu-plus}): writing $W=-\log G(T)$ and $F(\gamma)=\PR{W\le\gamma}$,
\begin{equation}\label{eq:awgn-rcuplus}
  \varepsilon_{\mathrm{ach}}^+(\log M)=\int_{\log M}^{\infty}F(\gamma)\,e^{\log M-\gamma}\,d\gamma,
\end{equation}
the only step from the random-coding ensemble being the Bernoulli envelope $1-(1-G(T))^{M-1}\le\min\BRAs{1,(M-1)G(T)}$. The two bounds share every other ingredient --- the same exactly computed function $G$ and the same law of $T$ --- so within this comparison the remaining difference between the curves is, by construction, the random-coding envelope.

\Cref{fig:awgn} brackets the optimal rate $R^*(n,\varepsilon)$ at $\mathrm{SNR}=0$~dB ($C=0.5$ bits/use), $\varepsilon=10^{-3}$, between the non-central-$t$ (NCT) meta-converse and the $\RCUp$ achievability. \Cref{tab:awgn} reports the reference operating point $n=200$. The bracket has width $0.009$ bits/use --- under $6\%$ of the finite-blocklength penalty $C-R^*\approx0.16$ bits --- and its measured widths ($0.053$ bits at $n=50$, $0.0016$ bits at $n=1000$) are consistent with closure at the dispersion rate $\Theta(1/\sqrt n)$. At the operating point the $\RCUp$ achievability ($R=0.3365$) lies $0.0092$ bits below the converse, whereas Polyanskiy's $\kappa\beta$ bound ($R=0.2837$) lies $0.0620$ bits below: the $\RCUp$ gap is about $6.7\times$ tighter. The relaxation cost of the common $Q_{Y^n}=\cN(0,(1+P)I_n)$ choice relative to the optimal output measure is only $0.003$ bits and appears to decay approximately as $1/n$ in these examples, an order faster than the bracket width.

\begin{table}[t]
\centering
\caption{Rate bounds for the AWGN channel at $n=200$, $\mathrm{SNR}=0$~dB, $\varepsilon=10^{-3}$ ($C=0.5000$ bits/use).}
\label{tab:awgn}
\begin{tabular}{lcc}
\toprule
Bound & $R$ (bits/use) & Gap to converse \\
\midrule
NCT converse (optimal meta-converse) & $0.3456$ & --- \\
$\chi^2$ converse \cite{polyanskiy2010channel} & $0.3484$ & --- \\
$\RCUp$ achievability & $0.3365$ & $0.0092$ \\
Normal approximation \cite{polyanskiy2010channel} & $0.3261$ & $0.0196$ \\
$\kappa\beta$ achievability \cite{polyanskiy2010channel} & $0.2837$ & $0.0620$ \\
\bottomrule
\end{tabular}
\end{table}

\begin{figure}[t]
  \centering
  \includegraphics[width=\columnwidth]{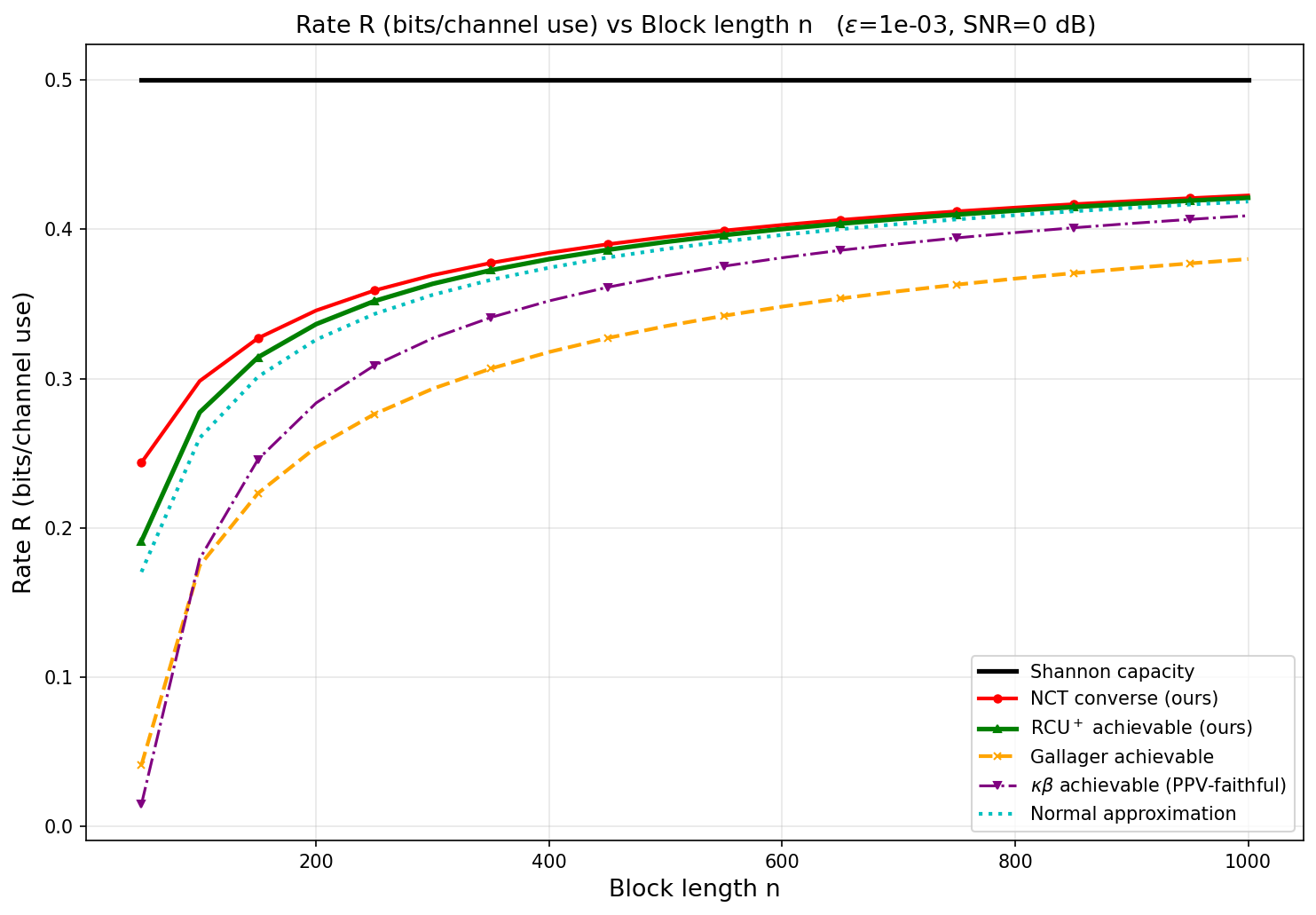}
  \caption{AWGN converse--achievability bracket at $\mathrm{SNR}=0$~dB, $\varepsilon=10^{-3}$, sphere-shell input. The non-central-$t$ meta-converse (\cref{thm:meta-converse}) and the $\RCUp$ achievability (\cref{thm:rcu-plus}) are driven by the same exactly evaluated pairwise kernel $G$; \cref{tab:awgn} reports the $n=200$ operating point.}
  \label{fig:awgn}
\end{figure}

\subsection{Binary Z-channel: the prior-family gap}

The second example isolates the value of optimizing the prior over the full type simplex rather than over memoryless laws. We use a binary Z-channel with crossover $z=0.1$,
\[
  \W=\begin{pmatrix}1 & 0\\ z & 1-z\end{pmatrix},\qquad C(\W)\approx0.763\text{ bits},
\]
chosen because its asymmetry makes the gap non-trivial: the capacity-achieving input is non-uniform, so a uniform memoryless prior is visibly suboptimal. On the type simplex the march of \cref{rem:ach-type} has one variable per composition --- $n+1=21$ types in the binary case --- and each iteration costs one gradient sweep over the type-level water-filling knots (\cref{prop:ach-march}); the converse-LP of \cref{thm:lp}, used below as a comparison point, is run in its type-reduced form (\cref{prop:lp-type}) --- $O(n^2)$ variables in this binary implementation --- and solves in seconds at this scale.

\Cref{fig:prioropt} compares, at $n=20$ under the $\RCUp$ kernel, the best achievability over the memoryless prior family against the best over the full type simplex, with the kernel held fixed so that only the prior family varies. The type-based optimum is computed by the simplex march of Section~\ref{sec:ach-prioropt}: Frank--Wolfe steps driven by the water-filling gradient of \cref{prop:ach-concave}, with exact line search along each step (the objective is piecewise-quadratic concave on a segment). It is independently verified: at every rate point the march value agrees with an exact convex-programming solve of the same program to within $4\times10^{-8}$, with the Frank--Wolfe duality gap as a self-contained optimality certificate.

In these experiments the march converges within a few hundred iterations per rate point, well under a second each, warm-started across the sweep. Enlarging the family from memoryless to type-based tightens the achievability across the rate range; the two curves are nearly indistinguishable on a $\log P_e$ axis, and the relative gain $(P_e^{\mathrm{ml}}-P_e^{\mathrm{march}})/P_e^{\mathrm{ml}}$ holds a plateau of about $3.4\%$ across the low-rate half of the sweep (peaking at $R\approx0.24$ bits) and vanishes at the high-rate end. The modest size of the gap under the exact $\RCUp$ kernel is itself informative: for this channel and blocklength, within the two prior families considered, the bulk of the achievability improvement available from prior optimization is already captured by the best memoryless law, with the type-based optimum contributing a residual low-to-mid-rate gain.

The experiment also shows that the achievability optimum genuinely requires its own program: the \emph{converse}-LP prior of \cref{thm:lp}, inserted into the $\RCUp$ bound as a surrogate, performs \emph{worse} than the best memoryless law at every rate in this sweep. The two prior optimizations share one water-filling geometry (\cref{rem:one-program}) but their optima differ, and only the march delivers the achievability-side one.

\begin{figure}[t]
  \centering
  \includegraphics[width=\columnwidth]{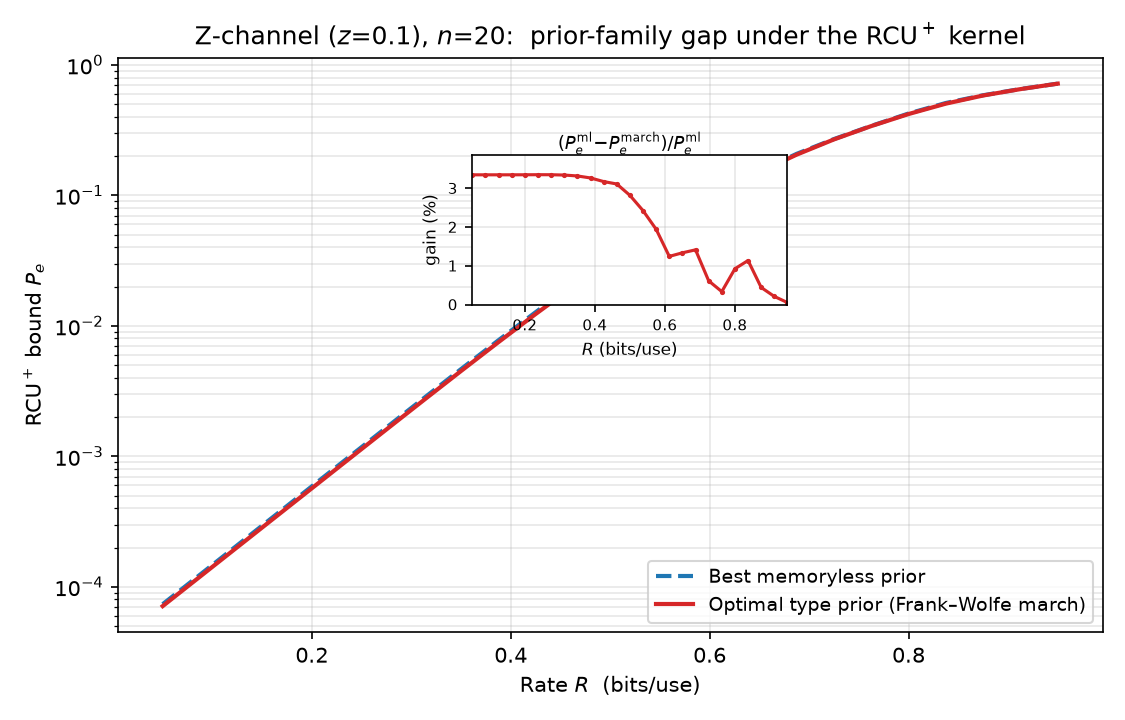}
  \caption{Prior-family gap for the binary Z-channel ($z=0.1$, $C\approx0.763$ bits) at $n=20$ under the $\RCUp$ kernel: best memoryless prior (dashed) versus the achievability-optimal type-based prior (solid), computed by the Frank--Wolfe march of \cref{prop:ach-march}. The curves overlap on $\log P_e$; the inset shows the relative gain $(P_e^{\mathrm{ml}}-P_e^{\mathrm{march}})/P_e^{\mathrm{ml}}$.}
  \label{fig:prioropt}
\end{figure}

\section{Conclusion}\label{sec:conc}

Three results summarize the paper.

\emph{The pairwise error probability depends on the decoding metric only
through the ordering it induces, providing a common representation for
arbitrary metrics.} Under the dithered tie rule the PEP obeys an exact
probability-integral-transform identity, and the error spectrum
$\Fspec{Q_X}{z}$ it induces is defined for an arbitrary decoding metric, at
any blocklength.

\emph{Achievability and converse are comparable through one spectrum.} The
random-coding bound (and its $\RCUp$ refinement) is a kernel integral of
$\Fspec{Q_X}{z}$, and the error probability of any fixed code is the same
spectrum evaluated at the code's empirical prior --- exactly, and for any
metric. The two directions differ only by a kernel and by the prior at which
the spectrum is read, and the comparison is quantitative: under log-concavity
of the spectrum their rate gap is at most $-\log(1-\varsigma)/\varsigma$,
which approaches one nat in the small-slope limit (\cref{thm:rategap}).

\emph{Both prior optimizations are computable.} Under matched ML the spectrum
is convex in the prior, and the reverse-channel identity turns the
prior-optimized minimax meta-converse into a finite-dimensional linear program
(\cref{thm:lp}) --- polynomial in the blocklength for memoryless channels with
fixed alphabets, after the type reduction of \cref{prop:lp-type}. The same
identity, integrated against the random-coding kernel, makes the
prior-optimized achievability bound a concave program over the input simplex
with an explicit gradient (\cref{prop:ach-concave,prop:ach-kkt}), optimized by
a first-order method with the convergence guarantee of \cref{prop:ach-march};
for Gallager-symmetric channels, a class strictly larger than the
transitive-group-symmetric one, the uniform prior is optimal in both directions
(\cref{prop:symmetric}).

Natural continuations --- joint source--channel and lossy source coding over
the same spectrum, refinements of the random-coding bound, and the
multi-terminal setting, where the bilinear coupling of encoder priors in the
meta-converse breaks convexity --- are left to future work.

\appendices

\section{Randomized Probability-Integral Transform and PEP Uniformity}\label{app:pep-uniform}

This appendix proves the uniformity and rank properties of the dithered PEP
(Lemma~\ref{thm:pep-uniform}). The engine is a randomized probability-integral
transform (RPIT) that restores uniformity in the presence of atoms.

\begin{lemma}[RPIT]\label{app:lem-rpit}
Let $T$ be a real random variable with CDF $F_T$ and atom function
$P_T(t)\triangleq F_T(t)-F_T(t^-)$, and let $U\sim\uU$ be independent of $T$.
Then $\Phi\triangleq F_T(T^-)+U\,P_T(T)\sim\uU$.
\end{lemma}

\begin{IEEEproof}
Work with the equivalent form $\Psi\triangleq F_T(T)-U\,P_T(T)$, related to
$\Phi$ by $U\mapsto1-U$ (which preserves $\uU$ and independence). Fix
$\phi\in(0,1)$.

\emph{Case 1: $F_T(t)=\phi$ for some $t$.} Split
$\PR{\Psi\le\phi}=\PR{\Psi\le\phi,T\le t}+\PR{\Psi\le\phi,T>t}$. If $T\le t$ then
$\Psi\le F_T(T)\le\phi$, so the first term is $\PR{T\le t}=\phi$. If $T>t$ then
$\Psi\ge F_T(T^-)\ge\phi$, with equality only when $P_T(T)=0$ and
$F_T(T^-)=\phi$, an event of probability $0$ (it forces $T$ into a flat piece of
$F_T$ above $t$); hence the second term vanishes and $\PR{\Psi\le\phi}=\phi$.

\emph{Case 2: no $t$ has $F_T(t)=\phi$.} By right-continuity there is a unique
$t_0$ with $F_T(t_0^-)\le\phi<F_T(t_0)$; set $\phi_1\triangleq F_T(t_0)$, so
$P_T(t_0)=\phi_1-F_T(t_0^-)>0$ and $\phi=\phi_1-u\,P_T(t_0)$ for a unique
$u\in(0,1]$. Then
$\PR{\Psi\le\phi}=\PR{\Psi\le\phi_1}-\PR{\phi<\Psi\le\phi_1}$. By Case~1,
$\PR{\Psi\le\phi_1}=\phi_1$. The event $\{\phi<\Psi\le\phi_1\}$ occurs only at
$T=t_0$, where $\Psi=\phi_1-U\,P_T(t_0)$, so it equals $\{T=t_0,\,U<u\}$ and has
probability $P_T(t_0)\,u=\phi_1-\phi$. Thus $\PR{\Psi\le\phi}=\phi$.
\end{IEEEproof}

\begin{IEEEproof}[Proof of Lemma~\ref{thm:pep-uniform}]
Parts (i)--(ii) are immediate from the definition~\eqref{eq:pep-def} of
$\PEPU{x}{y}{u}$ and of the order $\succ$. For (iii), fix $y$ and let
$T\triangleq m(X,y)$ with $X\sim Q_X$, CDF $F_T$ and atom function $P_T$. Writing
the strictly-greater mass as $G_{X,y}=1-F_T(T)$ and the tie mass as
$H_{X,y}=P_T(T)$,
\[
  \PEP{X}{y}=G_{X,y}+U\,H_{X,y}=1-\bigl(F_T(T)-U\,P_T(T)\bigr),
\]
which is $1-\Psi$ for the variable $\Psi$ of Lemma~\ref{app:lem-rpit}; since
$\Psi\sim\uU$, so is $\PEP{X}{y}$.
\end{IEEEproof}

The same construction, used over a finite candidate set, yields the rank
identity behind every fixed-code converse in the paper.

\begin{lemma}[Finite rank equivalence]\label{app:lem-rank}
Let $j\in[M]$ carry scores $d_j$ and i.i.d.\ dithers $U_j\sim\uU$, and let
$Q$ be the uniform prior on $[M]$. For an independent uniform index $J$, let
$\mathrm{rank}(J)$ be the position of $(d_J,U_J)$ among
$\{(d_j,U_j)\}_{j\in[M]}$ under $\succ$ (rank $1$ = top). With
$\PEPs{J}\triangleq G_J+U_J H_J$, $G_J=\#\{j:d_j>d_J\}/M$,
$H_J=\#\{j:d_j=d_J\}/M$, one has, for every $L\in\{1,\dots,M\}$,
\[
  \PR{\mathrm{rank}(J)\le L}=\PR{\PEPs{J}\le L/M}.
\]
\end{lemma}

\begin{IEEEproof}
Condition on the scores and on $J$, and put $a\triangleq\#\{j:d_j>d_J\}$ and
$k\triangleq\#\{j:d_j=d_J\}\ge1$ (the tie block contains $J$). Then
$\mathrm{rank}(J)=1+a+\#\{j\neq J:d_j=d_J,\,U_j<U_J\}$. Among the $k$ tied
dithers the rank of $U_J$ is uniform on $\{0,\dots,k-1\}$, so
$\PR{\mathrm{rank}(J)\le L}=\bigl(((L-a)_+)\wedge k\bigr)/k$. On the other side
$M\,\PEPs{J}=a+U_J k$, whence
$\PR{\PEPs{J}\le L/M}=\PR{U_J\le(L-a)/k}=\bigl(((L-a)_+)\wedge k\bigr)/k$, the
same value. Averaging over scores and $J$ proves the claim.
\end{IEEEproof}

\section{The Meta-Converse Identity}\label{app:meta-converse}

\begin{IEEEproof}[Proof of Theorem~\ref{thm:meta-converse}]
Write $\beta_\alpha(P;Q)=\min_{0\le T\le1,\,\Es{P}{T}\ge
\alpha}\Es{Q}{T}$ and abbreviate $Q^\bullet\triangleq(Q_X\times P_Y)\cdot m$.

\emph{($\le$).} The test $T_e(x,y)\triangleq\Ind{\PEP{x}{y}\ge e^{-R}}$ has, by
product-law uniformity (Lemma~\ref{thm:pep-uniform}(iii)),
\[
Q_X(T_e(X,y)=1)=1-e^{-R}\quad\text{at every }y,
\]
hence $(Q_X\times Q_Y)(T_e=1)=1-e^{-R}$ for every $Q_Y$; $T_e$ is therefore feasible
at level $\alpha=1-e^{-R}$, and
\[
\beta_{1-e^{-R}}(Q_X\times Q_Y;Q^\bullet)\le\Es{Q^\bullet}{T_e}
=\Es{Q_X\times P_Y}{m\,\Ind{Z_e\le R}}.
\]
Maximizing over $Q_Y$ gives ``$\le$''.

\emph{($\ge$).} For each $y$ pick $(\tau_y,\theta_y)$ solving the randomized
quantile equation $Q_X(m(\cdot,y)<\tau_y)+\theta_y\,Q_X(m(\cdot,y)=\tau_y)
=1-e^{-R}$. Set $\lambda^\ast\triangleq\sum_y P_Y(y)\tau_y$ and
$Q_Y^\star(y)\triangleq P_Y(y)\tau_y/\lambda^\ast$. From the quantile equation,
$\PR{m(X,y)\ge\tau_y}\ge e^{-R}$, so Markov's inequality gives
$\tau_y\le e^{R}\Es{Q_X}{m(X,y)}$ and hence, by assumption (A2),
$\lambda^\ast\le e^{R}\Es{Q_X\times P_Y}{m(X,Y)}<\infty$; thus $Q_Y^\star$ is a
probability measure (the degenerate case $m\equiv0$ is trivial). The
Neyman--Pearson likelihood ratio under $Q_Y^\star$ is
$L(x,y)=m(x,y)P_Y(y)/Q_Y^\star(y)=\lambda^\ast m(x,y)/\tau_y$, so the optimal
$\beta$-test at $Q_Y^\star$ accepts exactly where $m(\cdot,y)<\tau_y$, with tie
fraction $\theta_y$ at $m=\tau_y$ --- which is precisely $T_e$. Therefore
$\beta_{1-e^{-R}}(Q_X\times Q_Y^\star;Q^\bullet)=\Es{Q_X\times P_Y}{m\,
\Ind{Z_e\le R}}$, giving ``$\ge$''.
\end{IEEEproof}

\section{The Reverse-Channel Identity}\label{app:reverse-channel}

\begin{IEEEproof}[Proof of Theorem~\ref{thm:reverse-channel}]
The cap $D_\infty(W\|Q_X\mid Q_Y)\le z$ is the pointwise box
$0\le W(x\mid y)\le e^{z}Q_X(x)$ together with $\sum_x W(x\mid y)=1$. We prove the
$\sup$ form~\eqref{eq:rc-sup}, the form consumed by the linear program of
Section~\ref{sec:prioropt} and the water-filling of Section~\ref{sec:ach-prioropt};
the $\inf$ form~\eqref{eq:rc-inf} is the sign-reversed construction.

By uniformity (Lemma~\ref{thm:pep-uniform}), $\PRs{X\sim Q_X}{\PEP{X}{y}\le
e^{-z}}=e^{-z}$ at every $y$. Define
$W^\star(x\mid y)\triangleq e^{z}Q_X(x)\PR{\PEP{x}{y}\le e^{-z}}$ (the
probability taken over the dither). It sums to
$e^{z}\cdot e^{-z}=1$, satisfies $W^\star(x\mid y)\le e^{z}Q_X(x)$, hence is
feasible, and
$\Es{W^\star Q_Y}{m}=e^{z}\Es{Q_X\times Q_Y}{m\,\Ind{\PEP{X}{Y}\le e^{-z}}}$, so
the left side of~\eqref{eq:rc-sup} equals $e^{-z}\Es{W^\star Q_Y}{m}\le$ the
supremum.

For the reverse inequality fix $y$ and let $S_y\triangleq\{x:\PR{\PEP{x}{y}\le
e^{-z}}>0\}$. Since the PEP is decreasing in the rank of $m(\cdot,y)$, $S_y$
is a superlevel set: with $\mu_y\triangleq\inf_{x\in S_y}m(x,y)$ one has
$m(x,y)\ge\mu_y$ on $S_y$ and $m(x,y)\le\mu_y$ off $S_y$. By construction
$W^\star$ vanishes off $S_y$ and saturates the cap on $S_y$, except possibly on
the tie set where the dither probability is strictly between $0$ and $1$ ---
and there $m(x,y)=\mu_y$ exactly. For any feasible $W$ put
$\Delta(x)\triangleq W(x\mid y)-W^\star(x\mid y)$; feasibility forces
$\Delta\le0$ wherever $W^\star$ saturates the cap (where $m\ge\mu_y$) and
$\Delta\ge0$ off $S_y$ (where $m\le\mu_y$), while on the tie set $m=\mu_y$
regardless of the sign of $\Delta$; and $\sum_x\Delta(x)=0$.
Hence
\[
  \Es{W}{m(X,y)}-\Es{W^\star}{m(X,y)}=\sum_x\Delta(x)m(x,y)\le\mu_y\sum_x\Delta(x)=0,
\]
the inequality holding termwise ($\Delta\,m\le\Delta\,\mu_y$ in all three
regions). Thus $W^\star$ maximizes the objective at every $y$ ---
``fill from the top'': the budget-$1$ column is packed onto the largest metric
values first, at the maximal permitted density $e^{z}Q_X$; integrating over
$Q_Y$ yields~\eqref{eq:rc-sup}. The $\inf$ form replaces $S_y$ by the sublevel
set $\{x:\PR{\PEC{x}{y}\le e^{-z}}>0\}$ (``fill from the bottom'') and reverses
every inequality.
\end{IEEEproof}

The proof is written in countable notation to keep the bookkeeping light; in
the standard Borel setting of Section~\ref{sec:pep} it extends with no change
in structure. The cap $D_\infty(W\|Q_X\mid Q_Y)\le z$ forces
$W(\cdot\mid y)\ll Q_X$ with density at most $e^{z}$, so feasible channels are
identified with their densities; $W^\star$ is defined through its density
$e^{z}\PR{\PEP{x}{y}\le e^{-z}}$, measurable in $(x,y)$ by the joint
measurability of $(G_{x,y},H_{x,y})$; sums over $x$ become integrals against
$Q_X$; the exchange inequality $\Delta\,m\le\Delta\,\mu_y$ holds pointwise in
the same three regions; and the final average over $Q_Y$ is Fubini under the
integrability hypothesis of Theorem~\ref{thm:reverse-channel} (the analogue of
\textup{(A2)} with $Q_Y$ in place of $P_Y$).

\section{Slope Identity and the Fixed-Code Converse}\label{app:integral}

\begin{IEEEproof}[Proof of the slope identity~\eqref{eq:slope-identity}]
Write $P(R)\triangleq\tilde P_e(R;Q_X)$, $F(R)\triangleq\Fspec{Q_X}{R}$ and
$E(R)=-\log P(R)$. The spectrum $F$ is continuous (Lemma~\ref{lem:atomless}),
so differentiating the integral representation~\eqref{eq:integral} is
legitimate and gives
\[
  P'(R)=P(R)-F(R),
\]
whence $\dot E(R)=-P'(R)/P(R)=F(R)/P(R)-1$, which rearranges to
$P(R)=F(R)/(1+\dot E(R))$ at every $R$ with $F(R)>0$. For the slope bounds: the
integrand of~\eqref{eq:integral} satisfies $F(z)\ge F(R)$ for $z\ge R$
(monotonicity), so $P(R)\ge F(R)\,e^R\int_R^\infty e^{-z}dz=F(R)$, giving
$\dot E(R)\le0$; and $\dot E(R)=F(R)/P(R)-1>-1$ exactly when $F(R)>0$.
\end{IEEEproof}

\begin{IEEEproof}[Proof of Theorem~\ref{thm:cc-converse}]
Fix $y$ and let $Q_X$ be uniform over the $M$ codewords with scores
$d_i=m(x_i,y)$ and dithers $U_i$. By definition the dithered PEP of $x_i$ equals
$\PEPU{x_i}{y}{U_i}=G_{x_i,y}+U_i H_{x_i,y}=\PEPs{i}$ in the notation of
Lemma~\ref{app:lem-rank} (the empirical $G,H$ are exactly the normalized
strictly-greater and tie counts). A decoding error for the transmitted index
$I\sim\Unif{[M]}$ is the event $\mathrm{rank}(I)>1$, so by
Lemma~\ref{app:lem-rank} with $L=1$,
$\PRs{}{\text{error}\mid Y=y}=\PR{\mathrm{rank}(I)>1}=\PR{\PEP{X}{y}>1/M}$.
Averaging over $Y$ and using atomlessness (Lemma~\ref{lem:atomless}) to promote
$>$ to $\ge$ gives $P_e(\cC)=\PR{\PEP{X}{Y}\ge1/M}$.
\end{IEEEproof}

\section{The Rate-Gap Theorem}\label{app:rategap}

Throughout, $F$, $P$, $\varsigma=(\log F)'$, $R_F=F^{-1}(\varepsilon)$ and
$R_P=P^{-1}(\varepsilon)$ are as in \cref{sec:rategap}, under the hypotheses of
\cref{thm:rategap}; $P$ is non-decreasing, $F$ being so.

\begin{IEEEproof}[Proof of Theorem~\ref{thm:rategap}]
By the slope identity~\eqref{eq:slope-identity}, $P\ge F$ wherever $F>0$, so
$P(R_F)\ge F(R_F)=\varepsilon=P(R_P)$ and, $P$ being non-decreasing,
$R_F\ge R_P$. The upper half of the value sandwich (\cref{lem:sandwich} at
$R=R_P$, whose hypotheses are exactly those of the theorem) gives
$\varepsilon=P(R_P)\le F(R_P)/(1-\varsigma(R_P))$, i.e.\
$F(R_P)\ge\varepsilon(1-\varsigma(R_P))$. With $F\in C^1$ strictly increasing on
$[R_P,R_F]$ (its slope $F'=\varsigma F$ is positive there, since
$\varsigma\ge\varsigma(R_F)>0$) the change of variable $\rho=F(R)$ yields
\[
R_F-R_P=\int_{F(R_P)}^{\varepsilon}\frac{d\rho}{\varsigma(R(\rho))\,\rho}
\;\le\;\frac{1}{\varsigma(R_F)}\int_{\varepsilon(1-\varsigma(R_P))}^{\varepsilon}\frac{d\rho}{\rho}
=\frac{-\log\bigl(1-\varsigma(R_P)\bigr)}{\varsigma(R_F)},
\]
using $\varsigma(R)\ge\varsigma(R_F)$ on $[R_P,R_F]$ (non-increasing slope).
This is case~(i); when $\varsigma\equiv\varsigma_0$ on $[R_P,R_F]$ both
occurrences of the slope in the bound evaluate to $\varsigma_0$,
giving~\eqref{eq:rategap}.
\end{IEEEproof}

\section{The Local Pairwise Estimate}\label{app:pep-local}

The local sharpening invoked in \cref{rem:gauss-regime} is a lemma of the framework rather than an import: the change of measure that drives it is the Bayes posterior $P_{X\mid Y}$ --- the reverse-channel object the framework already works with --- rather than a tilting chosen for the proof. For blocks, write $\idens{x}{y}\triangleq\sum_{k=1}^{n}\idens{x_k}{y_k}$.

\begin{lemma}[Local pairwise estimate]\label{lem:pep-local}
Fix a product prior $Q_X^{\otimes n}$ on a memoryless channel $\W^{\otimes n}$ under the matched metric, and an output block $y=(y_1,\dots,y_n)$ with $P_Y^{\otimes n}(y)>0$. Under the product posterior $P_{X\mid Y}(\cdot\mid y)=\prod_{k=1}^{n}P_{X\mid Y}(\cdot\mid y_k)$, let $\sigma_k^2$ and $t_k$ denote the variance and the third absolute central moment of $\idens{X_k}{y_k}$, and set $V_n(y)\triangleq\sum_{k}\sigma_k^2$, $T_n(y)\triangleq\sum_{k}t_k$. If $V_n(y)\ge n\,v_0$ and $T_n(y)\le n\,\tau_0$ for some $v_0>0$ and $\tau_0<\infty$, then for every $x$ and every $u\in[0,1]$,
\begin{equation}\label{eq:pep-local}
\PEPU{x}{y}{u}\ \le\ \PRs{\bar X\sim Q_X^{\otimes n}}{\idens{\bar X}{y}\ge\idens{x}{y}}\ \le\ B\,\frac{e^{-\idens{x}{y}}}{\sqrt n},
\end{equation}
where
\[
B\triangleq\frac{1}{1-e^{-1}}\BRA{\frac{1}{\sqrt{2\pi v_0}}+\frac{2c_0\,\tau_0}{v_0^{3/2}}}
\]
and $c_0$ is an absolute (Berry--Esseen) constant; $c_0=6$ suffices~\cite[Ch.~XVI.5, Thm.~2]{feller1971vol2}.
\end{lemma}

\begin{IEEEproof}
Abbreviate $P\triangleq P_{X\mid Y}(\cdot\mid y)$, $Q\triangleq Q_X^{\otimes n}$, and
$a\triangleq\idens{x}{y}$. If $a=-\infty$ the right-hand side
of~\eqref{eq:pep-local} is infinite and there is nothing to prove; assume $a$
finite.

\emph{Step 1 (dither).} By Definition~\ref{def:pep}, for every $u\in[0,1]$,
\[
\PEPU{x}{y}{u}=G_{x,y}+u\,H_{x,y}\le G_{x,y}+H_{x,y}
=Q\BRAs{\bar x:m(\bar x,y)\ge m(x,y)}.
\]
The matched
metric is order-equivalent in $\bar x$, at each fixed $y$, to the information
density $\idens{\bar x}{y}$ (as in \cref{rem:gauss-regime}), so this set is
$\BRAs{\bar x:\idens{\bar x}{y}\ge a}$ --- the first inequality
of~\eqref{eq:pep-local}.

\emph{Step 2 (posterior change of measure).} For every block $\bar x$ with
$\W^{\otimes n}(y\mid\bar x)>0$,
\[
P(\bar x)=\frac{Q(\bar x)\,\W^{\otimes n}(y\mid\bar x)}{P_Y^{\otimes n}(y)}
=Q(\bar x)\,e^{\idens{\bar x}{y}},
\qquad\text{equivalently}\qquad
Q(\bar x)=P(\bar x)\,e^{-\idens{\bar x}{y}};
\]
blocks with $\W^{\otimes n}(y\mid\bar x)=0$ have $\idens{\bar x}{y}=-\infty$ and
lie outside the event $\{\idens{\bar x}{y}\ge a\}$. Hence, with
$S\triangleq\idens{\bar X}{y}$,
\[
\PRs{\bar X\sim Q}{S\ge a}=\Es{\bar X\sim P}{e^{-S}\,\Ind{S\ge a}}.
\]
Under $P$ the sum $S=\sum_{k}\idens{\bar X_k}{y_k}$ has independent (in general
non-identically distributed) terms, each finite $P$-almost surely, with total
variance $V_n(y)$ and third-moment sum $T_n(y)$.

\emph{Step 3 (unit-interval slicing).} Partitioning $\{S\ge a\}$ into the
slices $\{S\in[a+j,a+j+1)\}$, $j=0,1,2,\dots$, and using $e^{-S}\le e^{-(a+j)}$
on the $j$-th slice,
\[
\Es{P}{e^{-S}\,\Ind{S\ge a}}
\le\sum_{j\ge0}e^{-(a+j)}\,\PRs{P}{S\in[a+j,a+j+1)}
\le\frac{e^{-a}}{1-e^{-1}}\,\sup_{t\in\mR}\PRs{P}{S\in[t,t+1)}.
\]

\emph{Step 4 (Berry--Esseen small-ball bound).} Fix $t\in\mR$ and
$\epsilon>0$. Since $[t,t+1)\subset(t-\epsilon,t+1]$, applying the
Berry--Esseen theorem for independent, non-identically distributed
summands~\cite[Ch.~XVI.5, Thm.~2]{feller1971vol2} at both endpoints, and
bounding the increment of $\Phi$ by the peak normal density $1/\sqrt{2\pi}$
times the interval length,
\begin{align*}
\PRs{P}{S\in[t,t+1)}
&\le\Phi\!\BRA{\frac{t+1-\Es{P}{S}}{\sqrt{V_n(y)}}}
-\Phi\!\BRA{\frac{t-\epsilon-\Es{P}{S}}{\sqrt{V_n(y)}}}
+\frac{2c_0\,T_n(y)}{V_n(y)^{3/2}}\\
&\le\frac{1+\epsilon}{\sqrt{2\pi V_n(y)}}+\frac{2c_0\,T_n(y)}{V_n(y)^{3/2}}.
\end{align*}
Letting $\epsilon\downarrow0$ and inserting the hypotheses $V_n(y)\ge n v_0$,
$T_n(y)\le n\tau_0$,
\begin{align*}
\sup_{t\in\mR}\PRs{P}{S\in[t,t+1)}
&\le\frac{1}{\sqrt{2\pi V_n(y)}}+\frac{2c_0\,T_n(y)}{V_n(y)^{3/2}}\\
&\le\frac{1}{\sqrt{2\pi n v_0}}+\frac{2c_0\,n\tau_0}{(n v_0)^{3/2}}
=\frac{1}{\sqrt n}\BRA{\frac{1}{\sqrt{2\pi v_0}}+\frac{2c_0\,\tau_0}{v_0^{3/2}}}.
\end{align*}
Combining Steps 2--4 gives
$\PRs{\bar X\sim Q}{S\ge a}\le B\,e^{-a}/\sqrt n$ with $B$ as
in~\eqref{eq:pep-local}, the second inequality.
\end{IEEEproof}

The floor on $V_n(y)$ is where singularity enters. Writing $v(b)$ and $t(b)$ for the posterior variance and third absolute central moment of $\idens{X}{b}$ at a single output letter $b$, and $N(b\mid y)$ for the count of $b$ in $y$, one has $V_n(y)=\sum_b N(b\mid y)\,v(b)$ and $T_n(y)=\sum_b N(b\mid y)\,t(b)$. For finite alphabets $\tau_0=\max_b t(b)$ serves for every $y$; and some letter $b^\ast$ with $P_Y(b^\ast)>0$ has $v(b^\ast)>0$ exactly when the channel is non-singular on the support of $Q_X$ --- two inputs of positive prior mass with distinct positive likelihoods at a common output. In that case $V_n(y)\ge N(b^\ast\mid y)\,v(b^\ast)$, and a Chernoff--Hoeffding bound on the count gives the floor with $v_0=\tfrac12P_Y(b^\ast)\,v(b^\ast)$ for all output blocks outside a set of $P_Y^{\otimes n}$-probability at most $e^{-nP_Y(b^\ast)^2/2}$. For singular channels --- the binary erasure channel is the canonical case --- every $v(b)$ vanishes: at each output the posterior carries a single likelihood value, no $1/\sqrt n$ gain is available at any output, and correspondingly the $\tfrac12\log n$ term is absent from the BEC expansion.

\begin{remark}[What the refinement measures]\label{rem:pep-vs-is}
\Cref{lem:pep-local} also locates, quantitatively, where the information-spectrum
estimate loses accuracy. The bound $\PEP{x}{y}\le e^{-\idens{x}{y}}$ of
\cref{rem:gauss-regime} is the Chernoff bound on the pairwise tail at tilt
$s=1$: for every $s\ge0$,
$\PRs{\bar X\sim Q_X}{\idens{\bar X}{y}\ge a}\le
e^{-sa}\,\Es{Q_X}{e^{s\,\idens{\bar X}{y}}}$, and at $s=1$ the moment is
exactly one --- the choice information-spectrum arguments freeze. The exact
pairwise tail, which the PEP is, improves on this shadow in two separate
factors. First, the \emph{tilt}: away from the central regime the tail's true
exponent is attained at an optimizing $s\neq1$, which is precisely where
error-exponent analysis, with its optimized tilt
parameter~\cite{gallager1968information}, parts ways with the
information-spectrum bound; freezing $s=1$ is exponent-accurate only in the
dispersion regime. Second, the \emph{prefactor}: even at $s=1$ the local
central-limit factor $1/\sqrt n$ of \cref{lem:pep-local} is invisible to the
information-spectrum estimate, and it is the factor responsible for the
third-order improvement discussed in \cref{rem:gauss-regime}. The spectrum,
built from the exact PEP rather than from either bound, carries both
corrections automatically.
\end{remark}

\section{The LP Dual}\label{app:lp-dual}

We carry out the eliminate-and-substitute algebra behind
Remark~\ref{rem:lp-dual}. Attach a multiplier $\mu\in\mR$ to the prior
normalization~\eqref{eq:lp-c1}, a multiplier $\nu_y\in\mR$ to each per-output
normalization row in~\eqref{eq:lp-c2}, and a multiplier $\lambda_{xy}\ge0$ to
each box constraint~\eqref{eq:lp-c3}. The LP dual
of~\eqref{eq:lp-obj}--\eqref{eq:lp-c3} is
\begin{align*}
\text{minimize}\quad & \mu+\sum_{y}\nu_y\\
\text{subject to}\quad & \nu_y+\lambda_{xy}\ \ge\ e^{-R}\,\W(y\mid x),\quad
\lambda_{xy}\ge0,\quad\forall x,y,\\
& \mu\ \ge\ e^{R}\sum_{y}\lambda_{xy},\qquad\forall x\in\cX,
\end{align*}
the first constraint family dual to the variables $W^*_{X\mid Y}(x\mid y)$ and
the second to the variables $Q_X(x)$. Substitute $t_y\triangleq e^{R}\nu_y$ and
eliminate $\lambda_{xy}$ at its smallest feasible value
$\lambda_{xy}=e^{-R}\max\{\W(y\mid x)-t_y,\,0\}$; one checks that
$t_y\in[0,1]$ is without loss at the dual optimum (coordinates $t_y<0$ or
$t_y>1$ can only increase the objective). Using $\sum_y\W(y\mid x)=1$, the dual
objective becomes
\begin{align*}
\mu+\sum_y\nu_y
&=\max_{x\in\cX}\Bigl(1-\sum_y\min\{\W(y\mid x),t_y\}\Bigr)+e^{-R}\sum_y t_y\\
&=1-\min_{x\in\cX}\Bigl[\sum_y\min\{\W(y\mid x),t_y\}-e^{-R}\sum_y t_y\Bigr].
\end{align*}
Minimizing over $t\in[0,1]^{|\cY|}$ and invoking LP strong duality against the
primal value $1-\inf_{Q_X}\Fspec{Q_X}{R}$ of \cref{thm:lp}
gives~\eqref{eq:lp-dual}, with the optimal multipliers of the per-output rows
$\nu_y^\star=e^{-R}t_y^\star$. Finally, complementary slackness on the prior
rows ($Q_X^\star(x)>0\Rightarrow\mu=e^{R}\sum_y\lambda_{xy}$) states that every
input in the support of the optimal prior attains the inner minimum
of~\eqref{eq:lp-dual} --- the converse-side counterpart of the
equal-marginal-value condition of \cref{prop:ach-kkt}.

\section{Water-Filling and the Simplex March}\label{app:ach-prioropt}

\begin{IEEEproof}[Proof of Lemma~\ref{lem:waterfill}]
The sup form of the reverse-channel identity~\eqref{eq:rc-sup} reads, for this
metric and the conditional $D_\infty$ cap at level $z=-\log w$,
\[
G(Q_X;w)=\sup_{W^{*}:\,W^{*}(x\mid y)\le w^{-1}Q_X(x)}w\sum_{x,y}W^{*}(x\mid y)\W(y\mid x),
\]
with $W^{*}$ a stochastic reverse channel (the auxiliary $Q_Y$ cancels). The
change of variable $V(x)\triangleq w\,W^{*}(x\mid y)$ turns the cap into
$V\le Q_X$, stochasticity into $\sum_x V(x)=w$, and the objective into
$\sum_x V(x)\W(y\mid x)$; the constraints decouple across $y$, giving
\eqref{eq:waterfill}. For the per-output value, decompose the objective by
levels: $\nu^{y}_{j}=\sum_{k\ge j}c^{y}_{k}$, so for every feasible $V$,
\[
\sum_x V(x)\,\W(y\mid x)
=\sum_{k}c^{y}_{k}\sum_{j\le k}V\bigl(x^{y}_{(j)}\bigr)
\;\le\;\sum_{k}c^{y}_{k}\,\min\bigl\{w,\sigma^{y}_{k}\bigr\},
\]
each partial sum being at most the budget $w$ and at most the box mass
$\sigma^{y}_{k}$. The water-filling
$V\bigl(x^{y}_{(i)}\bigr)=\min\bigl\{Q_X(x^{y}_{(i)}),(w-\sigma^{y}_{i-1})_+\bigr\}$
makes every partial sum equal $\min\{w,\sigma^{y}_{k}\}$ simultaneously, so the
bound is attained and \eqref{eq:wf-ramp} holds; tied inputs share a common
$\nu^{y}_{j}$, so the value is tie-break independent. Differentiating
\eqref{eq:wf-ramp} in $w$ gives \eqref{eq:wf-staircase}:
$\partial_w\min\{w,\sigma^{y}_{k}\}=\Ind{w<\sigma^{y}_{k}}$ and
$\sum_{k:\,\sigma^{y}_{k}>t}c^{y}_{k}=\nu^{y}_{j}$ for
$t\in(\sigma^{y}_{j-1},\sigma^{y}_{j})$ by telescoping. Every stated property is
read off \eqref{eq:wf-ramp}: each ramp $\min\{w,\sigma^{y}_{j}\}$ is
non-decreasing, concave and piecewise-linear in $w$, its breakpoint
$\sigma^{y}_{j}$ is linear in $Q_X$, and the weights $c^{y}_{j}$ are nonnegative.
\end{IEEEproof}

\begin{IEEEproof}[Proof of the gradient formula~\eqref{eq:ach-grad} of
Proposition~\ref{prop:ach-concave}]
Substitute the staircase $s_y(Q_X;w)=\int_0^w\rho_y(Q_X;t)\,dt$
of~\eqref{eq:wf-staircase} into $J$ and exchange the order of integration,
\[
\int_0^1 s_y(Q_X;w)\,\kappa(w)\,dw
=\int_0^1\rho_y(Q_X;t)\,\bar\kappa(t)\,dt
=\sum_{j}\nu^{y}_{j}\!\int_{\sigma^{y}_{j-1}}^{\sigma^{y}_{j}}\!\bar\kappa(t)\,dt .
\]
Write $m=j(x,y)$ and $B(u)\triangleq\int_0^u\bar\kappa$, so the per-output term is
$\sum_j\nu^{y}_{j}\bigl[B(\sigma^{y}_{j})-B(\sigma^{y}_{j-1})\bigr]$. The prior mass
$Q_X(x)$ enters $\sigma^{y}_{k}$ exactly for $k\ge m$, with
$\partial\sigma^{y}_{k}/\partial Q_X(x)=\Ind{k\ge m}$, and $B'=\bar\kappa$; hence
\[
\frac{\partial}{\partial Q_X(x)}\sum_j\nu^{y}_{j}\bigl[B(\sigma^{y}_{j})-B(\sigma^{y}_{j-1})\bigr]
=\sum_{j\ge m}\nu^{y}_{j}\,\bar\kappa(\sigma^{y}_{j})
-\sum_{j\ge m+1}\nu^{y}_{j}\,\bar\kappa(\sigma^{y}_{j-1})
=\sum_{j\ge m}\bigl(\nu^{y}_{j}-\nu^{y}_{j+1}\bigr)\,\bar\kappa(\sigma^{y}_{j}),
\]
the last equality reindexing the second sum by $j\mapsto j+1$. Summing over
$y$ gives~\eqref{eq:ach-grad}.
\end{IEEEproof}

\begin{IEEEproof}[Proof of Proposition~\ref{prop:ach-march}]
The cumulative masses are linear in the prior, so
\[
\abs{\sigma^{y}_{j}(Q_X)-\sigma^{y}_{j}(Q_X')}\le\norm{Q_X-Q_X'}_{1},
\]
and $\bar\kappa$ is $\kappa(0^{+})$-Lipschitz because $\bar\kappa'=-\kappa$ with
$0\le\kappa\le\kappa(0^{+})$. Using
$\sum_{j\ge m}(\nu^{y}_{j}-\nu^{y}_{j+1})=\nu^{y}_{m}$ with $m=j(x,y)$ and
$\nu^{y}_{j(x,y)}=\W(y\mid x)$ in~\eqref{eq:ach-grad},
\begin{align*}
\bigl|g(x;Q_X)-g(x;Q_X')\bigr|
&\le\kappa(0^{+})\,\norm{Q_X-Q_X'}_{1}\sum_{y}\nu^{y}_{j(x,y)}\\
&=\kappa(0^{+})\,\norm{Q_X-Q_X'}_{1}\sum_{y}\W(y\mid x)
=L\,\norm{Q_X-Q_X'}_{1},
\end{align*}
which is~\eqref{eq:ach-smooth}. The telescoping step is metric-agnostic; the
matched hypothesis enters only in identifying the telescoped weight
$\nu^{y}_{j(x,y)}$ with the channel value $\W(y\mid x)$, whose sum over outputs
is $1$ --- so the constant $L=e^{R}$ is specific to the matched metric. Hence $-J$ is $L$-smooth with respect to
$\norm{\cdot}_{1}$, whose dual norm is $\norm{\cdot}_{\infty}$; the simplex has
$\norm{\cdot}_{1}$-diameter $2$, so the Frank--Wolfe curvature constant is at most
$L\cdot2^{2}=4L$, and the standard guarantee
$\max_{Q_X}J-J(Q_X^{(k)})\le 2C_f/(k+2)$ for an $L$-smooth concave objective
yields~\eqref{eq:ach-rate}.
\end{IEEEproof}

\bibliographystyle{IEEEtran}
\bibliography{refs}

\end{document}